\documentclass[12pt]{article}
\usepackage{graphicx}
\usepackage{enumerate}
\usepackage{natbib}
\usepackage{url} 
\usepackage{bm}
\usepackage[caption=false]{subfig}
\usepackage{rotating}
\usepackage{xcolor}
\usepackage{placeins}
\usepackage[
    colorlinks=true,      
    linkcolor=red,       
    citecolor=blue,       
    urlcolor=blue,        
    breaklinks=true,      
]{hyperref}

\usepackage{mathtools}
\usepackage{xr}

\usepackage{verbatim}

\usepackage[export]{adjustbox}
\usepackage{array}

\usepackage{enumitem}
\usepackage{setspace}
\usepackage{multirow, booktabs}
\usepackage{amsmath, amsthm, amssymb, amsfonts}

\usepackage{color}

\newcommand{\rev}[1]{#1}
\usepackage[T1]{fontenc}
\usepackage{mathtools, nccmath}
\usepackage{tabularx}

\usepackage[ruled, vlined]{algorithm2e}

\usepackage{algpseudocode}

\usepackage{float}
\usepackage{caption}

\usepackage{comment}
\usepackage{pdflscape}
\usepackage{pifont}

\SetKwInOut{Input}{Input}
\SetKwInOut{Output}{Output}

\allowdisplaybreaks

\setlist[enumerate]{itemsep=0mm}

\DeclareMathAlphabet\mathbfcal{OMS}{cmsy}{b}{n}

\DeclarePairedDelimiter{\nint}\lfloor\rceil

\newcommand\numberthis{\addtocounter{equation}{1}\tag{\theequation}}
\def\code#1{\texttt{#1}}
\allowdisplaybreaks

\makeatletter
\newcommand{\vast}{\bBigg@{4}}
\newcommand{\Vast}{\bBigg@{10}}
\makeatother

\newcommand{\bmx}{\begin{bmatrix}}
\newcommand{\emx}{\end{bmatrix}}

\def\ve{\varepsilon}
\def\bve{\boldsymbol{\varepsilon}}

\def\b{\boldsymbol}

\def\CS#1{\texttt{\textbackslash#1}}

\DeclarePairedDelimiter{\sfloor}{\big\lfloor}{\big\rfloor}
\DeclarePairedDelimiter{\sceil}{\big\lceil}{\big\rceil}

\DeclarePairedDelimiter{\ceil}{\Big\lceil}{\Big\rceil}
\DeclarePairedDelimiter{\floor}{\Big\lfloor}{\Big\rfloor}

\DeclareMathOperator*{\argmin}{arg\,min}

\newcommand{\dbtilde}[1]{\tilde{\tilde{#1}}}

\newcommand{\lvertiii}[1]{{\left\vert\kern-0.25ex\left\vert\kern-0.25ex\left\vert #1 
    \right\vert\kern-0.25ex\right\vert\kern-0.25ex\right\vert}}
\newcommand{\vertiii}[1]{{\vert\kern-0.25ex\vert\kern-0.25ex\vert #1 
    \vert\kern-0.25ex\vert\kern-0.25ex\vert}}

\theoremstyle{plain}
\newtheorem{thm}{Theorem}

\newtheorem{lem}{Lemma}

\theoremstyle{definition}

\theoremstyle{remark}

\newtheorem{Thm}{\bf Theorem}[section]

\newif\ifJASA
\JASAfalse

\newif\ifnotblind
\notblindtrue 

\ifJASA
\else

\fi

\def\spacingset#1{\renewcommand{\baselinestretch}%
{#1}\small\normalsize} \spacingset{1}

{
  \title{\bf High-dimensional sparsity-adaptive multiple change-point detection}
\ifnotblind
  \author{  Hyeyoung Maeng\\
  Department of Mathematical Sciences, Durham University.\\
  Department of Statistics, Ewha Womans University.\footnote{Second affiliation effective 1 September 2026.}
  \\~\\ 
  Tengyao Wang\\
  Department of Statistics, London School of Economics.
  \\~\\
  Piotr Fryzlewicz\\
  Department of Statistics, London School of Economics.
  }
\fi
}

\begin{document}

\maketitle


\begin{abstract}
We introduce a method for detecting multiple change-points in the mean of a high-dimensional data sequence. Unlike existing top-down (i.e.\ divisive) algorithms, we adopt a bottom-up (i.e. agglomerative) approach, whereby we iteratively merge \rev{neighboring} segments of data starting from the finest level. This is particularly useful for signals with frequent change-points, since local evidence is assessed before segments are combined into coarser summaries. We compute $L_2$- and $L_\infty$-aggregated test statistics of \rev{neighboring} segments and combine the information from their respective ranks, which makes the method adaptive in handling different degrees of change-point sparsity. We show the consistency of the estimated number and locations of change-points under both iid Gaussian and possibly dependent and/or non-Gaussian noise. The practicality of our approach is demonstrated through simulations and a real data example involving the UK House Price Index data.
\ifnotblind
Our methodology is implemented in the R package \code{BUHDA}, available at \url{https://github.com/hmaeng/BUHDA}.
\fi
\end{abstract}

\noindent%


\section{Introduction}\label{sec:intro}

High-dimensional data arise in many different fields including finance, environmental science, biology, economics and astronomy. When vast quantities of data are collected over time, the data-generating mechanism may experience change; this paper focuses on the situation in which such changes are of an abrupt nature. 
Recent examples of such applications include detection of exoplanets from light curve data \citep{fisch2022linear}, detecting the most recent change-point in a telecommunications network  \citep{bardwell2019most}, detecting forest changes using satellite images \citep{morresi2024high}, detecting changes in functional magnetic resonance imaging (fMRI) data for a set of subjects \citep{cribben2017estimating, li2019change}, detecting price inflation from UK retail price indices \citep{groen2013multivariate} and detecting changes in the incidence of terrorism \citep{tickle2021computationally}.

In this paper, we consider a panel of $p$ univariate data sequences recorded over $n$ time steps, where the dimension $p$ and the sequence length $n$ may both be large, and $p$ may be comparable with, or even larger than, $n$. We study the following data generation model for high-dimensional panel data,
\begin{equation} \label{model}
X_{i, t} = f_{i, t} + \ve_{i, t},\quad  i=1, \ldots, p, \quad t=1, \ldots, n,
\end{equation}
where $\b f_{i,\cdot} = (f_{i, 1}, \ldots, f_{i, n})^\top$ is the underlying signal vector of the $i$th component time series $\b X_{i,\cdot} = (X_{i, 1}, \ldots, X_{i, n})^\top$. 
Initially, we consider the case that the innovations $\varepsilon_{i,t} \sim N(0,\sigma^2)$ are independent across $i$ and $t$. The assumption of temporal independence is relaxed in Appendix \ref{appendixB}, in which we consider possibly dependent and/or non-Gaussian noise. 
We keep the assumption of cross-sectional independence throughout the paper and show in Section \ref{sec4.5} how PCA can be used to deal with possible cross-sectional dependence in practical situations.
We assume that change-points in the signal vectors $\{\b f_{i,\cdot}\}_{i=1}^p$, if present, have locations within the set $\{\eta_1, \eta_2, \ldots, \eta_N\}$, where
\begin{equation} \label{e5.1.2}
0 < \eta_1 < \eta_2 < \ldots < \eta_N < n.
\end{equation}
We set $\eta_0=0$ and $\eta_{N+1}=n$ by convention. The value of $N$ is unknown and can grow with $n$. 
At each change-point $\eta_\ell$, we assume that there exists at least one coordinate (e.g. $k$th data sequence) at which $f_{k, \eta_\ell}$ and $f_{k, \eta_{\ell}+1}$ differ, and  
 $\{\b f_{i,\cdot}\}_{i=1}^p$ is otherwise constant between any adjacent change-points, in the following way:
\begin{align*} \numberthis \label{e5.1.3}
& \Omega_\ell=\big\{i\in\{1,\ldots,p\}: \big|f_{i, \eta_{\ell}+1}-f_{i, \eta_{\ell}} \big| \neq 0 \big\} \neq \emptyset \;\text{ for } \; \ell=1, \ldots, N \\
& \text{ where $f_{i, t}= \theta_{i, \ell}$ for $t\in[\eta_{\ell-1}+1, \eta_{\ell}],  \; \ell=1, \ldots, N+1, \; i=1,\ldots, p$}.
\end{align*}
For each change-point $\eta_\ell$, change can occur in a dense subset of the signal components (e.g. all or most components of $\{f_{i,\eta_\ell}\}_{i=1, \ldots, p}$) or only in a sparse subset of the components, where the level of sparsity is described by $\mathcal{S}_\ell \coloneqq |\Omega_\ell|$.

Numerous methods for high-dimensional change-point analysis have been proposed, see e.g.\ \citet{bai2010common}, \citet{zhang2010detecting}, \citet{horvath2012change}, \citet{wang2022inference}, \citet{jirak2015uniform}, \citet{yu2021finite}, \citet{cho2015multiple}, \citet{cho2016change}, \citet{wang2018high}, \citet{chen2022inference}. Explicitly considering adaptivity to varying sparsity in high-dimensional change-point analysis is still relatively underexplored in the literature and only a small number of methods have been proposed (e.g. \citet{enikeeva2019high}, \citet{liu2020unified}, \citet{zhang2022adaptive}, \citet{wang2023computationally}).

In this paper, we introduce a sparsity-adaptive bottom-up algorithm for detecting multiple change-points in the mean of a high-dimensional data sequence, which we refer to as `BUHDA' (Bottom-Up High-Dimensional Adaptive change-point detection). In contrast to the divisive top-down approach, our bottom-up approach agglomeratively merges adjacent data segments that are least likely to contain change-points. In the univariate setting, \citet{fryzlewicz2017tail} and \citet{maeng2023detecting} demonstrate the attraction of the bottom-up framework for detecting multiple change-points. 
They provide empirical evidence that the bottom-up approach works well in detecting frequent change-points including abrupt local features where many existing top-down change-point detection methods fail.
The current work extends these ideas to the high-dimensional setting. The bottom-up methodology reduces the construction of the change-point solution path to repeated pairwise comparison of two high-dimensional vectors. This is useful when changes are frequent and their sparsity levels vary: short segments are tested before they are merged, and combining the $L_2$ and $L_\infty$ ranks lets the method detect both dense and sparse changes without fixing the sparsity level in advance.
More precisely, our main contributions are as follows.

\begin{enumerate}
    \item The bottom-up tree construction starts with the finest level of data, where each data point $\b X_{\cdot, t}$ is treated as its own segment with corresponding node $(t-1, t]$. In this context, merge refers to the process of combining adjacent nodes into higher-level parent nodes. Our first contribution is adaptivity achieved in the bottom-up tree construction. To deal with possibly varying sparsity over change-points, we use both the $L_2$ and $L_\infty$ norms of CUSUM-type statistics in deciding which \rev{neighboring} regions should be merged next. In aggregating both the $L_2$ and $L_\infty$ norms, we first obtain the rank vector of each norm ($R^{L_2}$ and $R^{L_\infty}$) by sorting each norm for all possible merges, then find the combined rank by taking the entrywise maximum of $R^{L_2}$ and $R^{L_\infty}$. We merge the pair of \rev{neighboring} segments whose combined rank is the smallest so that we postpone the merge of a segment if it has either a dense or a sparse change (thus the $L_2$ or the $L_\infty$ norm is large enough, respectively). We remark that such rank-based combination of different across-panel aggregators of contrast statistics is only possible in bottom-up/agglomerative approaches. 
    \item Bottom-up methods target local features at an early stage, before focusing on more global features corresponding to longer data segments. 
    Therefore, they tend to perform better than top-down approaches in estimating the number of change-points in frequent change-point scenarios, although they tend to underperform in \rev{localization} i.e. estimating the locations of change-points \citep{maeng2023detecting}. 
    This is because the initial merges are based on  short segments of the data.
    To improve \rev{localization}, we add pre-merging and adjusting steps to the bottom-up merge algorithm. 
    Pre-merging adds stability to the algorithm by ensuring that even the finest-level merging tests are performed on large enough sample sizes. Adjusting adds flexibility to the algorithm by making it less greedy. We expand on these aspects in Section~\ref{improving_le}.
\end{enumerate}
As demonstrated in Section \ref{sec4}, the above ingredients lead to good performance of BUHDA in scenarios in which the sparsity of change varies over change-points and when relatively frequent change-points exist, especially in high-dimensional settings.

This paper is organized as follows. 
In Section \ref{sec2}, we give a full description of the BUHDA procedure and Section \ref{sec3} presents the relevant theoretical results.
The supporting numerical studies including a real house-price data example are given in Section \ref{sec4}. 
The proofs of our main theoretical results and theoretical extensions to non-Gaussian and/or dependent noise are in the Appendix.

\section{Methodology}\label{sec2}

\subsection{Bottom-up tree construction} \label{btutree}

In contrast to the top-down methods, bottom-up procedures start from the finest level of the data and iteratively merge the most similar neighboring pairs of segments until all data points are in the same segment. The tree construction algorithm consists of several merge passes through the data, alongside auxiliary pre-merge and adjust passes. The pre-merge and adjust passes will be discussed in Section~\ref{improving_le}. We first focus on the merge passes, which form the backbone of the entire procedure. Initially, each data point $\b X_{\cdot, t}$ is its own segment, corresponding to a node $(t-1, t]$ of the bottom-up tree. In our algorithm, a \emph{node} is an interval $(a, b]:=\{a+1, \ldots, b\}$ of the domain $\{1, \ldots, n\}$ of the time series, together with the history of its construction, where the history is achieved through storing descendant nodes using the function $\mathrm{children}(\cdot)$ defined in Algorithm \ref{algo_mergepass}. In each pass, the algorithm considers each pair of neighboring segments to decide which pairs to merge first. The merging priority is determined via a dissimilarity statistic, so that the most similar segments are merged first. For two neighboring segments $(u,v]$ and $(v,w]$, the statistic is computed by aggregating over the panel the following componentwise CUSUM statistic: 
\begin{align} \label{cusum}
C_{i; u, v, w} = \sqrt{\frac{(w-v)(v-u)}{w-u}} \big( \bar{X}_{i, (u,v]} - \bar{X}_{i, (v,w]} \big), \; i=1, \ldots, p,
\end{align}
where $u < v < w$ and $\bar{X}_{i, (a,b]} = 1/(b-a) \sum_{t=a+1}^b X_{i, t}$. In each single pass, a proportion $\rho$ of neighboring segment pairs get merged and the resulting segments can be thought of as parents of the constituent children segments, leading to a tree-like construction. Consequently, the algorithm needs at most a logarithmic number of passes through the data to construct the root node of the tree, in which all data points belong to the same segment.

\subsection{Adaptivity to unknown sparsity level}\label{L2Linfty}
We now illustrate how BUHDA achieves adaptivity in the bottom-up tree construction to deal with possibly varying sparsity $\mathcal{S}_\ell$. We first formulate the coordinate-wise $L_2$ and $L_\infty$ aggregations of the CUSUM statistics as follows:
\begin{eqnarray} 
\label{l_2}  C^{L_2}_{u, v, w} &=& \bigg\{\sum_{i=1}^p (C_{i; u, v, w})^2 \bigg\}^{1/2},  \\ 
\label{l_infty}  C^{L_\infty}_{u, v, w} &=& \max_{i\in\{1,\ldots,p\}} \big|C_{i; u, v, w} \big|. 
\end{eqnarray}
We now provide a simple example to illustrate how the data transformation is carried out using both the $L_2$ and $L_\infty$ norms.
\paragraph{Example.}
Suppose that $\rho=0.3$, and the input data matrix of the dimension $3 \times 6$ is,
\begin{align*}
\b X = \begin{pmatrix} 
X_{1, 1} & X_{1, 2}  & X_{1, 3}  & X_{1, 4}  & X_{1, 5}  & X_{1, 6} \\
X_{2, 1} & X_{2, 2}  & X_{2, 3}  & X_{2, 4}  & X_{2, 5}  & X_{2, 6} \\
X_{3, 1} & X_{3, 2}  & X_{3, 3}  & X_{3, 4}  & X_{3, 5}  & X_{3, 6} \\
\end{pmatrix}.
\end{align*}
As shown in the right diagram of Figure \ref{fig:btutree}, from the initial input of the data, there exist six nodes, $\{(0, 1], (1, 2], (2, 3], (3, 4], (4, 5], (5, 6]\}$, in the current layer (i.e. $L=6$) with the corresponding segments, $\{\bar{\b X}_{\cdot, (0, 1]}, \bar{\b X}_{\cdot, (1, 2]}, \bar{\b X}_{\cdot, (2, 3]}, \bar{\b X}_{\cdot, (3, 4]}, \bar{\b X}_{\cdot, (4, 5]}, \bar{\b X}_{\cdot, (5, 6]}\}$, where $L$ is the number of current layer nodes, $\bar{\b X}_{\cdot, (a, b]} = (\bar{X}_{1, (a, b]}, \ldots, \bar{X}_{p, (a, b]})^\top$ and $\bar{X}_{i, (a, b]} = 1/(b-a) \sum_{t=a+1}^b X_{i, t}$. A layer represents a depth in a tree (i.e. distance from the bottom); therefore a horizontal row of nodes sharing the same depth belongs to the same layer. \\
Pass $j=1$: We initially look at all pairwise differences by computing the vectors of corresponding CUSUM statistics in \eqref{cusum},
\begin{align*}
\begin{pmatrix} 
C_{1;0,1,2}\\
C_{2;0,1,2}\\
C_{3;0,1,2} \\
\end{pmatrix},
\begin{pmatrix} 
C_{1;1,2,3}\\
C_{2;1,2,3}\\
C_{3;1,2,3} \\
\end{pmatrix},
\begin{pmatrix} 
C_{1;2,3,4}\\
C_{2;2,3,4}\\
C_{3;2,3,4} \\
\end{pmatrix},
\begin{pmatrix} 
C_{1;3,4,5}\\
C_{2;3,4,5}\\
C_{3;3,4,5} \\
\end{pmatrix},
\begin{pmatrix} 
C_{1;4,5,6}\\
C_{2;4,5,6}\\
C_{3;4,5,6} \\
\end{pmatrix},
\end{align*}
and this gives rise to five possible candidate merges. Then we compute the $L_2$ and $L_\infty$ aggregations of the CUSUM statistics as in \eqref{l_2} and \eqref{l_infty} for each of the five candidate merges and their corresponding ranks as follows
\begin{align*}
   L_2: \Big(C^{L_2}_{0,1,2}, C^{L_2}_{1,2,3}, C^{L_2}_{2,3,4}, C^{L_2}_{3,4,5}, C^{L_2}_{4,5,6}\Big) \rightarrow \Big(R^{L_2}_{0,1,2}, R^{L_2}_{1,2,3}, R^{L_2}_{2,3,4}, R^{L_2}_{3,4,5}, R^{L_2}_{4,5,6}\Big), \\ \numberthis \label{Ranks}
   L_\infty: \Big(C^{L_\infty}_{0,1,2}, C^{L_\infty}_{1,2,3}, C^{L_\infty}_{2,3,4}, C^{L_\infty}_{3,4,5}, C^{L_\infty}_{4,5,6}\Big) \rightarrow \Big(R^{L_\infty}_{0,1,2}, R^{L_\infty}_{1,2,3}, R^{L_\infty}_{2,3,4}, R^{L_\infty}_{3,4,5}, R^{L_\infty}_{4,5,6}\Big),
\end{align*}
where $R^{L_2}_{\cdot, \cdot, \cdot}$ and $R^{L_\infty}_{\cdot, \cdot, \cdot}$ are ranks of $C^{L_2}_{\cdot, \cdot, \cdot}$ and $C^{L_\infty}_{\cdot, \cdot, \cdot}$, respectively. We then combine the above two rank vectors to form
\begin{align*}
\Big(R^*_{0,1,2}, R^*_{1,2,3}, R^*_{2,3,4}, R^*_{3,4,5}, R^*_{4,5,6}\Big),
\end{align*}
where 
\begin{equation}\label{Rstar}
    R^*_{\cdot, \cdot, \cdot} = \max \big(R^{L_2}_{\cdot, \cdot, \cdot}, R^{L_\infty}_{\cdot, \cdot, \cdot}\big).
\end{equation}
We then merge the pairs whose corresponding $R^*$'s are the smallest. Intuitively, if $C^{L_2}$ or $C^{L_\infty}$ is large enough, then merging the corresponding pair is postponed and highly likely to survive until a later stage of merging. This allows us to identify the change-points by looking at nodes nearest to the root in the tree construction. By way of illustration, suppose that $\big(R^*_{0,1,2}, R^*_{1,2,3}, R^*_{2,3,4}, R^*_{3,4,5}, R^*_{4,5,6}\big) = (2, 4, 5, 5, 2)$. As $\sceil{\rho L} = 2$, we pick the first two smallest ranks, $R^*_{0,1,2}$ and $R^*_{4,5,6}$, and merge the corresponding pairs. Once the merges are done, the current layer node set becomes $\{(0, 2], (2, 3], (3, 4], (4, 6]\}$ with the corresponding segments, $\{\bar{\b X}_{\cdot, (0, 2]}, \bar{\b X}_{\cdot, (2, 3]}, \bar{\b X}_{\cdot, (3, 4]}, \bar{\b X}_{\cdot, (4, 6]}\}$, as illustrated in Figure \ref{fig:btutree}. 
\begin{figure}[ht] 
\centering
\includegraphics[width=13cm, height=5.5cm]{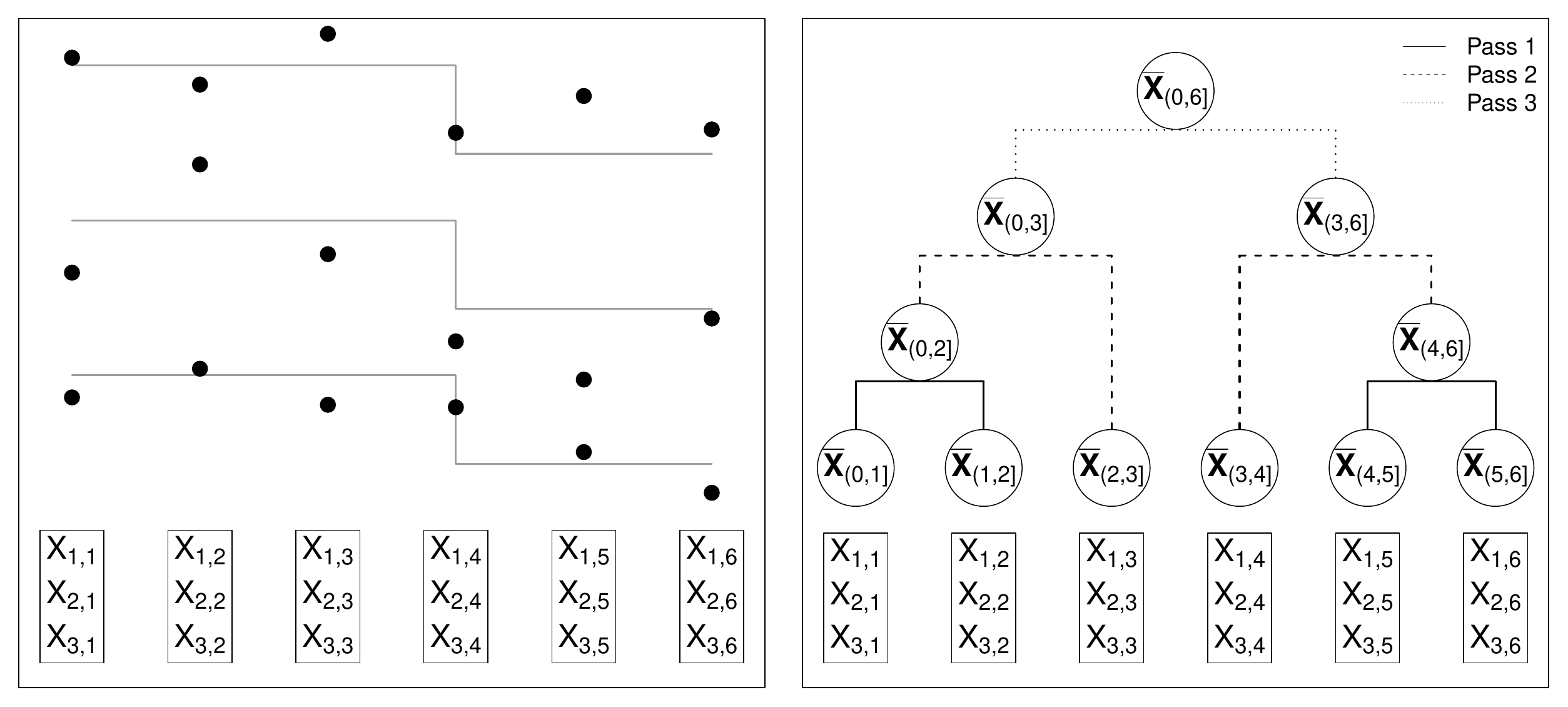}
 \caption{Example of bottom-up tree construction presented in Section \ref{btutree}. Left diagram: observations $X$ (dots) and underlying signals $f$ (solid line). Right diagram: bottom-up binary tree with segments (circled). Notation simplified in this figure: $\bar{\b X}_{(a, b]} \coloneqq \bar{\b X}_{\cdot, (a, b]}$.}
\label{fig:btutree}
\end{figure} 
\\
Pass $j=2$: We now have four nodes in the current layer ($L=4$) so that $\sceil{\rho L} = 2$. Using the current layer nodes, $\{(0, 2], (2, 3], (3, 4], (4, 6]\}$, compute the $L_2$ and $L_\infty$ of the CUSUM statistics and the corresponding combined ranks,
\begin{align*}
   R^*_{0,2,3}, R^*_{2,3,4}, R^*_{3,4,6},
\end{align*}
and this gives rise to three possible candidate merges. Suppose $R^*_{0,2,3}$ and $R^*_{3,4,6}$ are the two smallest: merge the corresponding pairs and update the current layer node set as $\{(0, 3], (3, 6]\}$. \\
Pass $j=3$: We now have only one pair to merge and the algorithm ends with one current layer node, $\{(0, 6]\}$.

We refer to all three passes in this example as ``merge passes'' (see Algorithm \ref{algo_mergepass}). Merge passes are the key ingredient of the bottom-up tree construction and interact with other operations such as pre-merge passes and adjust passes; the details of those will be given in Section \ref{improving_le}.

\begin{algorithm}[h!]
\SetAlgoLined
\DontPrintSemicolon
\KwIn{$\mathbf{X}_{p \times n}$, $\rho$ (merge function parameter), $c_\text{pm}$ (pre-merge pass parameter)}
\KwOut{$H$ (merging history sequence)}
$O \leftarrow \{(0, 1],(1, 2], (2, 3], \ldots, (n-1, n]\}$ \\
$H \leftarrow \langle O \rangle$ \hspace{0.5cm} \tcp*{Initialise history sequence}

\For{$j = 1$ \KwTo $\left\lfloor \log_2(n) / c_\text{pm} \right\rfloor$}{
    $O \leftarrow \texttt{pre-merge}(O)$ 
    \hspace{2.3cm} \tcp*[r]{See Algorithm~\ref{algo_premergepass}}
    
    $O \leftarrow \texttt{adjust}(O)$ 
    \hspace{3.15cm} \tcp*[r]{See Algorithm~\ref{algo_adjustpass}}

    $H \leftarrow H \Vert \langle O \rangle$
    \hspace{1cm} \tcp*{Concatenate current state to history}
}

\While{$|O| \geq 2$}{
    $O \leftarrow \texttt{merge}(O, \rho)$ 
    \hspace{2.15cm} \tcp*[r]{See Algorithm~\ref{algo_mergepass}}

    $O \leftarrow \texttt{adjust}(O)$

     $H \leftarrow H \Vert \langle O \rangle$
}

\Return{$H$}

\caption{Pseudocode for bottom-up tree construction}
\label{algo_bu}
\end{algorithm}

\begin{algorithm}[ht!]
\SetAlgoLined
\DontPrintSemicolon
\Input{
parameter $\rho$,\\
a list of nodes $O= \{(o_0, o_1], (o_1, o_2], \ldots, (o_{L-1}, o_L]\}$
}
\Output{
updated list of nodes $O$
}
\SetKwFunction{FMain}{merge}
  \SetKwProg{Fn}{Function}{:}{\KwRet $O$}
  \Fn{\FMain{$O, \rho$}}{ 
   Define $\b C^{L_2} = (C^{L_2}_{o_k, o_{k+1}, o_{k+2}})_{0\leq k\leq L-2}$ and $\b C^{L_\infty} = (C^{L_\infty}_{o_k, o_{k+1}, o_{k+2}})_{0\leq k\leq L-2}$ following \eqref{l_2} and \eqref{l_infty}.\;
  Compute the corresponding rank vectors $\b R^{L_2}$ and $\b R^{L_\infty}$ as in \eqref{Ranks}. \;
  Combine $\b R^{L_2}$ and $\b R^{L_\infty}$ to form a vector $\b{R}^{*}$ as in \eqref{Rstar} and let $r_1,\ldots,r_{\lceil \rho L\rceil}$ be the indices of the $\lceil \rho L\rceil$ smallest elements of $\b{R}^*$ (in increasing order)\;
  \For{$j$ in $1,\ldots,\lceil \rho L\rceil$}{
      If $\{(o_{r_j}, o_{r_j+1}], (o_{r_j+1}, o_{r_j+2}]\} \subseteq O$, then merge to form a new node $(o_{r_j}, o_{r_j+2}]$ such that $\text{children}((o_{r_j}, o_{r_j+2}]) := \{(o_{r_j}, o_{r_j+1}], (o_{r_j+1}, o_{r_j+2}]\}$ and update 
      \[    
      O \leftarrow \{(a,b] \in O: \text{$b\leq o_{r_j}$ or $a \geq o_{r_j+2}$}\} \cup (o_{r_j}, o_{r_j+2}].
      \]
      }
}
\caption{Pseudocode for merge pass}
\label{algo_mergepass}
\end{algorithm}

As shown in the above example, the merging process is repeated until all data points are in the same segment. This leads to the construction of a solution path which arranges the change-point candidates in the order of importance. This hierarchy justifies the use of thresholding for change-point estimation, described in Section \ref{thresholding}.

\subsection{Improving localization error} \label{improving_le}
If merge passes are only used in constructing a bottom-up tree, the algorithm tends to underperform in localization as the bottom part of the merge tree is built by focusing on local features identified with relatively short segments. The novelty of our BUHDA procedure includes adding other types of operations: pre-merge passes and adjust passes to improve localization. The algorithm for the bottom-up tree construction including those three types of passes is formulated in Algorithm \ref{algo_bu} and their details are given in the following (for completeness, we also include merge passes below).

\paragraph{Merge passes.} The merge passes are formally defined in Algorithm~\ref{algo_mergepass} and illustrated with an example in Section~\ref{L2Linfty}.

\paragraph{Pre-merge passes.}\label{sec:premerge_pass}

\begin{figure}[ht] 
\centering
\includegraphics[width=13cm, height=5.5cm]{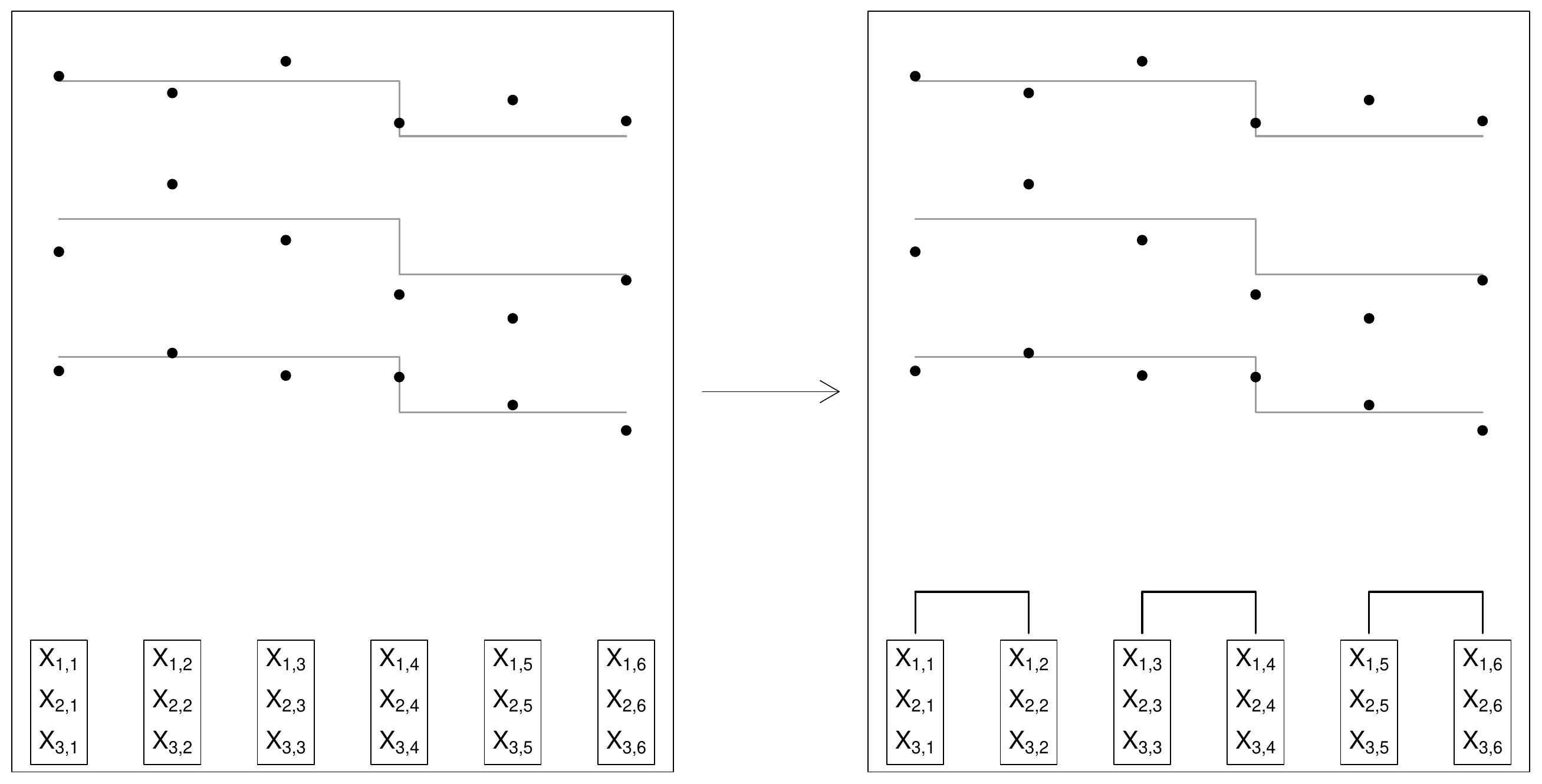}
 \caption{Example of pre-merge pass}
\label{fig:premerge1}
\end{figure}

\begin{algorithm}[h!]
\SetAlgoLined
\DontPrintSemicolon
\Input{
a list of nodes $O= \{(o_0, o_1], (o_1, o_2], \ldots, (o_{L-1}, o_{L}]\}$
}
\Output{
updated list of nodes $O$
}
\SetKwFunction{FMain}{pre-merge}
  \SetKwProg{Fn}{Function}{:}{\KwRet $O$}
  \Fn{\FMain{$O$}}{ 
  We set 
  \[
  O \leftarrow \begin{cases} 
  \{(o_0, o_2],\ldots,(o_{L-2}, o_{L}]\} & \text{if $L$ is even}\\
  \{(o_0, o_2],\ldots, (o_{L-3}, o_{L-1}], (o_{L-1}, o_{L}]\} & \text{if $L$ is odd and $o_1-o_0 < o_{L}-o_{L-1}$}\\
  \{(o_0, o_1], (o_1, o_3],\ldots, (o_{L-2}, o_{L}]\} & \text{if $L$ is odd and $o_1-o_0 \geq o_{L}-o_{L-1}$},
  \end{cases}
  \]
  and $\mathrm{children}((o_j, o_{j+2}]) := \{(o_j, o_{j+1}], (o_{j+1}, o_{j+2}]\}$ for each $(o_j, o_{j+2}]\in O$.
}
\caption{Pseudocode for pre-merge pass}
\label{algo_premergepass}
\end{algorithm}

Before performing any merge passes, we coalesce neighboring segments to create longer segments of length $O(n^{1/c_{\mathrm{pm}}})$. This is achieved by running the pre-merge pass, as specified in Algorithm~\ref{algo_premergepass}, $\lfloor \log_2 n / c_{\mathrm{pm}}\rfloor$ times. 
 The goal of the pre-merge phase is to guarantee that sufficiently long segments are used in the computation of the initial CUSUM statistics. Figure~\ref{fig:premerge1} illustrates the effect of replacing the first merge pass with a pre-merge pass in the example of Section~\ref{btutree}.

The number of pre-merge passes affects the shape of tree and thus the estimated change-points. If there is no pre-merge pass, then the initial stage of tree construction can be affected by outliers. On the other hand, if there are too many pre-merge passes compared to merge passes, the shape of tree becomes less data-adaptive as pre-merge pass does not use the CUSUM statistics. The choice of the parameter $c_{\mathrm{pm}}$, which determines the number of pre-merge pass, will be discussed in Section \ref{sec4.1}.

\begin{figure}[ht] 
\centering
\includegraphics[width=13cm, height=5.5cm]{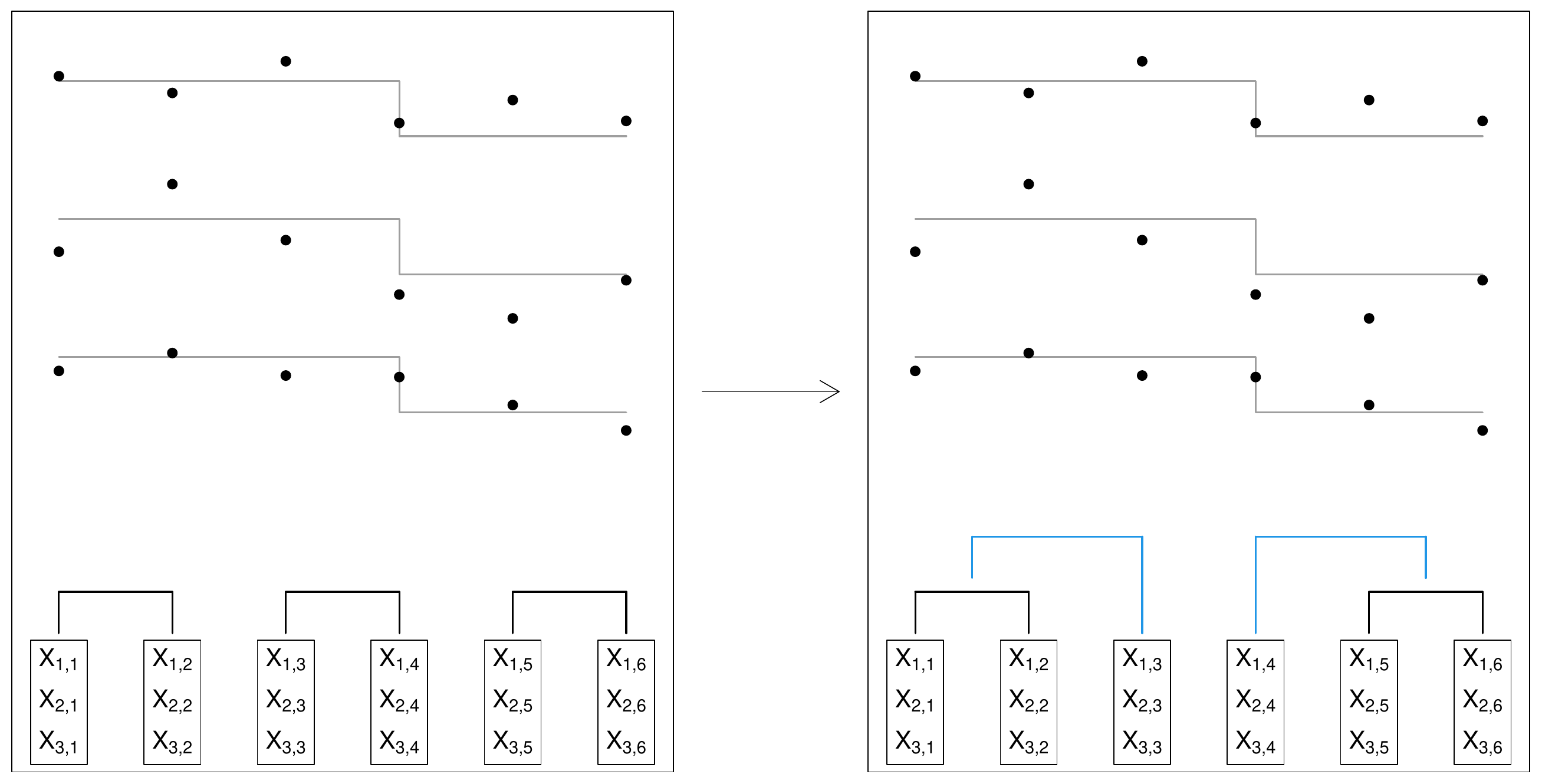}
 \caption{Example of adjust pass}
\label{fig:adjust1}
\end{figure} 

\paragraph{Adjust passes.} \label{sec:adjust_pass}
The adjust pass provides an opportunity to reverse the effects of the pre-merge or merge passes by selectively reassigning the splits where local evidence suggests better segmentation. For this purpose, the adjust pass follows every pre-merge and merge passes and examines possible nodes in the current layer. Adjusting a node is equivalent to splitting a node into two children nodes, which are then respectively combined with its left and right current-layer neighbors. Such an operation decreases the total number of segments by 1. The adjust pass is done in a conservative way in that we only split a current layer node if both its left and right children are respectively more similar to the left and right neighboring segments in the current layer, according to both $L_2$ and $L_\infty$ aggregated CUSUM statistics. As an illustration, Figure \ref{fig:adjust1} shows the example of an adjust pass following a pre-merge pass, which occurs as a result of the following:
\begin{align*}
    \max\big\{C^{L_2}_{0,2,3}, C^{L_2}_{3,4,6}\big\} < C^{L_2}_{2,3,4} \; \text{ and } \; 
    \max\big\{C^{L_\infty}_{0,2,3}, C^{L_\infty}_{3,4,6}\big\} < C^{L_\infty}_{2,3,4}. 
\end{align*}
The formal algorithm for an adjust pass can be found in Algorithm \ref{algo_adjustpass}.

\begin{algorithm}[ht!]
\SetAlgoLined
\DontPrintSemicolon
\Input{
A list of nodes $O= \{(o_0, o_1], (o_1, o_2], \ldots, (o_{L-1}, o_{L}]\}$, \\
}
\Output{
updated list of nodes $O$
}
\SetKwFunction{FMain}{adjust}
  \SetKwProg{Fn}{Function}{:}{\KwRet $O$}
  \Fn{\FMain{$O$}}{
    \For{$k$  in $0, \ldots, L-1$}{
    Let $o'_k$ be defined such that $\text{children}((o_{k}, o_{k+1}]) = \{(o_k, o'_k], (o'_k, o_{k+1}]\}$. \;
    For $q\in\{2,\infty\}$,  
    \begin{align*}
        D_k^{L_q, \text{left}} &:= C^{L_q}_{o_k, o_{k}', o_{k+1}} - C^{L_q}_{o_{k-1}, o_{k}, o_{k}'}\\
        D_k^{L_q, \text{right}} &:= C^{L_q}_{o_k, o_{k}', o_{k+1}} - C^{L_q}_{o_k', o_{k+1}, o_{k+2}},
    \end{align*}    
    where by convention $D_k^{L_q, \text{left}} = \infty$ if $k=0$ and $D_k^{L_q, \text{right}} = \infty$ if $k = L-1$.\;
    }
    Define $\mathcal{K} = \{k:  D^{\mathrm{min}}_k > 0\}$ where
    \begin{align*}
      D^{\text{min}}_k =\min \big\{ p^{-1/2}D_k^{L_2, \text{left}}, p^{-1/2}D_k^{L_2, \text{right}}, D_k^{L_\infty, \text{left}}, D_k^{L_\infty, \text{right}}\big\}.  
    \end{align*}
    \For{$k$ in $\mathcal{K}$ (\text{sorted by $D^{\mathrm{min}}_k$ in decreasing order)}}{
      If $\{(o_{k-1}, o_{k}], (o_{k}, o_{k+1}], (o_{k+1}, o_{k+2}]\} \subset O$, then unmerge $(o_{k}, o_{k+1}]$ into its children nodes and merge them with their left and right neighbor nodes to form a set of new nodes $\{(o_{k-1}, o_{k}'], (o_{k}', o_{k+2}]\}$ and update
      \[    
      O \leftarrow \{(a,b] \in O: \text{$b\leq o_{k-1}$ or $a \geq o_{k+2}$}\} \cup \{(o_{k-1}, o_{k}'], (o_{k}', o_{k+2}]\}.
      \]
      }
  
}
\caption{Pseudocode for adjust pass}
\label{algo_adjustpass}
\end{algorithm}

\subsection{Signal estimation via thresholding} \label{thresholding}

As mentioned above, our procedure constructs a bottom-up tree via Algorithm~\ref{algo_bu}. In this section, we describe how this tree is used in detecting change-points and estimating underlying signal vectors.  
We use thresholding as a way of deciding the significance of the obtained CUSUM statistics and thereby detect change-points based on the bottom-up tree. For each non-leaf node in the tree, the boundary between its two children nodes defines a change-point candidate. For a node $(u, w]$ with children nodes $\{(u, v], (v, w]\}$, we declare a change-point at $v$ if either  $C^{L_2}_{u, v, w}$ or $C^{L_\infty}_{u, v, w}$ exceeds its respective pre-specified threshold, $\lambda_{2}$ and $\lambda_{\infty}$, which classifies the initial estimated change-points into three categories: $L_{2}$, $L_\infty$, and $L_2+L_\infty$. An example can be found in Section \ref{sec4.5}. In addition, we also classify a node as having a change if any of its descendants have a change. We refer to this construction as \emph{connected thresholding}, since this ensures that the set of significant nodes form a connected pruned tree, which ensures consistent signal reconstruction from a theoretical point of view \citep[see, e.g.][for more details]{fryzlewicz2017tail}. 

After pruning, each surviving node corresponds to an estimated change-point, separating the node's children. Sorting these change-point locations in increasing order as $\tilde\eta_1,\ldots,\tilde\eta_{\tilde N}$, we obtain the initial estimators ${\tilde{\b f}}$ for the signal vectors as follows:
\begin{equation} \label{ftilde}
\tilde{f}_{i, t} =  \frac{1}{\tilde{\eta}_\ell-\tilde{\eta}_{\ell-1}} \sum_{s=\tilde{\eta}_{\ell-1}+1}^{\tilde{\eta}_\ell} X_{i,s} \quad \text{for} \quad t \in \big[\tilde{\eta}_{\ell-1}+1, \tilde{\eta}_\ell\big], \; i=1, \ldots, p, \; \ell=1, \ldots, \tilde{N}+1,  
\end{equation}
where $\tilde{\eta}_{0}=0$ and $\tilde{\eta}_{\tilde{N}+1}=n$.

\subsection{Additional considerations} \label{additional}

\subsubsection{Post-processing for consistency of change-point detection}

As will be shown in Theorem \ref{thm1} in Section \ref{sec3}, the piecewise-constant estimator ${\tilde{\b f}}$ in \eqref{ftilde} is consistent in the $L_2$ sense, which implies no underestimation of the number of change-points. On the other hand, empirical evidence suggests it may overestimate the number of change-points. We thus apply a two stage post-processing framework similar to that of \citet{fryzlewicz2017tail} to remove spurious estimated change-points so as to achieve consistency in estimating both the number and location of the change-points. 
More specifically, we post-process the estimated change-points in the following two stages: 

\paragraph{Stage 1.} \label{pp1}
In this stage, we re-run merge passes using the estimated mean ${\tilde{\b f}}$ from \eqref{ftilde} as our input data, but merging only one node in each pass. In other words, in each merge pass, we merge the two neighboring nodes with smallest combined $L_2$ and $L_\infty$ contrasts in the sense described in Section~\ref{L2Linfty}. These merge passes are performed until either of the $L_2$ and $L_\infty$ contrast exceeds their corresponding thresholds ($\lambda_{2}$ and $\lambda_{\infty}$ respectively), at which point, the current layer nodes in the bottom-up tree form a segmentation of the data. Say that there are $\dbtilde{N}+1$ current layer nodes, this give rise to $\dbtilde{N}$ change-points $(\dbtilde{\eta}_1,\ldots,\dbtilde{\eta}_{\dbtilde{N}})$. We can also form a new estimator ${\dbtilde{\b f}}$ for the mean signal as follows,
\begin{equation} \label{fdtilde}
\dbtilde{f}_{i, t} =  \frac{1}{\dbtilde{\eta}_\ell-\dbtilde{\eta}_{\ell-1}} \sum_{s=\dbtilde{\eta}_{\ell-1}+1}^{\dbtilde{\eta}_\ell} X_{i,s} \quad \text{for} \quad t \in \big[\dbtilde{\eta}_{\ell-1}+1, \dbtilde{\eta}_\ell\big], \; i=1, \ldots, p, \; \ell=1, \ldots, \dbtilde{N}+1,  
\end{equation}
where $\dbtilde{\eta}_{0}=0$ and $\dbtilde{\eta}_{\dbtilde{N}+1}=n$.

\paragraph{Stage 2.} \label{pp2}
This stage prunes the change-points $(\dbtilde{\eta}_1, \dbtilde{\eta}_2, \ldots, \dbtilde{\eta}_{\dbtilde{N}})$ in ${\dbtilde{\b f}}$ to obtain final estimators of the change-points and the mean signal. For each $\ell=1, \ldots, \dbtilde{N}$, we compute both $C^{L_2}_{u_\ell, v_\ell, w_\ell}$ and $C^{L_\infty}_{u_\ell, v_\ell, w_\ell}$ aggregated CUSUM statistics as in \eqref{l_2} and \eqref{l_infty}, respectively, setting $u_\ell=\floor{\frac{\dbtilde{\eta}_{\ell-1} + \dbtilde{\eta}_\ell}{2}}$, $v_\ell=\dbtilde{\eta}_\ell$ and $w_\ell=\ceil{\frac{\dbtilde{\eta}_{\ell} + \dbtilde{\eta}_{\ell+1}}{2}}$. We then find the minimizer of the combined rank, $\ell_0 = \argmin_\ell{R^*_{u_\ell, v_\ell, w_\ell}}$. If the following conditions are satisfied,
\begin{equation} \label{cond}
    C^{L_2}_{u_{\ell_0}, v_{\ell_0}, w_{\ell_0}} \leq \lambda_{2} \text{ and } C^{L_\infty}_{u_{\ell_0}, v_{\ell_0}, w_{\ell_0}} \leq \lambda_{\infty},
\end{equation}
we remove $\dbtilde{\eta}_{\ell_0}$ and repeat the above pruning process, until no further change-point can be removed. We write $\hat{N}$ for the number of detected change-points after this final pruning stage and let $\hat\eta_1,\ldots,\hat\eta_{\hat N}$ be the remaining estimated change-points in increasing order (also, by convention, we set $\hat{\eta}_{0}=0$ and $\hat{\eta}_{\hat{N}+1}=n$).

The estimated function $\hat{\b f}$ is obtained as in \eqref{fdtilde}, with $\hat\eta_{\ell}$ replacing $\dbtilde{\eta}_{\ell}$ therein. Through these two stages of post processing, consistency of the estimated number and locations of change-points can be achieved. The corresponding theoretical results can be found in Section~\ref{sec3}.

\subsubsection{Computational complexity}
The bottom-up tree construction in Algorithm \ref{algo_bu} consists of three passes: pre-merge, adjust and merge passes, each of which, applied to a current layer with $m$ nodes, has a computational complexity of order $O(pm+m\log(n))$ to account for both CUSUM calculation and ordering of $L_2$ and $L_\infty$ statistics. There can be at most $J = \lceil \log(n) / \log(1/(1-\rho))\rceil$ merge and pre-merge passes, since each of them reduces the total number of nodes at least by a multiple of $1-\rho$ and the adjust pass that interlaces them also reduces number of nodes by 1. Hence, the overall complexity of the bottom-up tree construction is of order $O(pn + n\log n)$. The post-processing steps have a complexity of $O(\tilde N n)$, where $\tilde N$ is the number of change-points identified in the bottom-up tree. In view of the probabilistic bound on $\tilde N$ from Theorem~\ref{thm1}, the worst-case complexity for the post-processing step is of order $O(Nn\log n)$ with high probability.

\section{Theoretical results}\label{sec3}

The theoretical results stated in this section consider the i.i.d.\ Gaussian noise. The details for dependent, possibly non-Gaussian, noise are in Appendix~\ref{appendixB}. 

We first study the $L_2$ consistency of $\big\{{\tilde{\b f}}_i\big\}_{i=1}^p$ and $\big\{{\dbtilde{\b f}}_i\big\}_{i=1}^p$, and then the change-point estimation consistency of $\big\{\b{\hat{f}}_i\big\}_{i=1}^p$, where the estimators are defined in Section \ref{sec2}. 
The $L_2$ risk of $\big\{{\tilde{\b f}}_i\big\}_{i=1}^p$ is defined as $\big\| {\tilde{\b f}} - \b{f} \big\|_{p, n}^2 = \frac{1}{pn} \sum_{i=1}^p \sum_{t=1}^n \big( \tilde{f}_{i, t}-f_{i, t} \big)^2$, where $\b{f}_i$ is the underlying signal in \eqref{model}.

\begin{thm} \label{thm1}
Suppose that $\{\b{X}_{i, \cdot}\}_{i=1}^p$ follow model \eqref{model} with $\sigma_i=1$ for all $i=1, \ldots, p$. Assume that $p \lesssim n^\alpha$ for some fixed $\alpha \in (0, \infty)$, there exist constants $c_1, c_2 >0$ such that taking $\lambda_\infty = c_1\log^{1/2}(n)$ and $\lambda_{2} = c_2 \sqrt{p + \log n}$, we have on an event with probability approaching 1 as $n\to\infty$ that  \begin{equation} \label{thm1-eq}
\begin{split}
\| {\tilde{\b f}}-\b{f} \|_{p, n}^2 \; \leq \; \frac{1}{n} \Bigg[ & \min\Big(c_1^2 \log(n), c_2^2\frac{p+\log n}{p}\Big) \\
& + 4N \biggl\lceil\frac{\log(n)}{\log (1/(1-\rho))}\biggr \rceil \max_\ell \bigg\{ c_1^2 \frac{\mathcal{S}_\ell\log (n)}{p} \wedge c_2^2\frac{p+\log n}{p} \bigg\} \Bigg].
\end{split}
\end{equation}
On the same event, the piecewise-constant estimator $\{ {\tilde{\b f}}_{i, \cdot} \}_{i=1}^p$ contains $\tilde{N} \leq C N \log (n)$ change-points for an absolute constant $C >0$.
\end{thm}

Thus, ${\tilde{\b f}}$ is $L_2$ consistent if $\frac{N\log^2n}{n} =o(1)$. 
The second term in~\eqref{thm1-eq} shows how the estimation accuracy adapts to the sparsity level $\mathcal{S}_\ell$, $\ell=1,\ldots,N$, for the $N$ change signals. Specifically, this term is linear in $\max_\ell \mathcal{S}_\ell$ up to $\max_\ell \mathcal{S}_\ell \asymp p/\log(n)$ and does not depend on the maximum sparsity level when the order of the latter exceeds $p/\log(n)$.  

We now look into the property of the estimator ${\dbtilde{\b f}}$ obtained after the first stage of post-processing.

\begin{thm} \label{thm2}
Under the assumptions of Theorem \ref{thm1}, we have $\big\| {\dbtilde{\b f}} - \b{f} \big\|_{p, n}^2 = O(R_{p, n})$ with probability approaching $1$ as $n \rightarrow \infty$, where
\begin{equation}\label{rpn}
R_{p, n} = N\frac{\log n}{n} \max_\ell \bigg\{ \frac{\mathcal{S}_\ell \log n}{p} \wedge \Big(1+\frac{\log n}{p}\Big) \bigg\}.
\end{equation}
Moreover, there exist at most two estimated change-points in $(\eta_\ell, \eta_{\ell+1}]$ for each $\ell=0, \ldots, N$; in particular, $\dbtilde{N} \leq 2(N+1)$.
\end{thm}
We see that ${\dbtilde{\b f}}$ is $L_2$ consistent for estimating the mean signal. While it still possibly overestimates the number of change-points, Theorem~\ref{thm2} shows that we can control the spurious changes to be at most 2 on each stationary segment. This allows us to remove these spurious change-points in another post-processing stage. This second stage of the post-processing is designed to achieve consistency in estimating both the number and location of the change-points.

\begin{thm} \label{thm3}
We assume that the assumptions given in Theorem \ref{thm1} hold with an additional condition $N=O(\log n)$ and let $R_{p, n}$ be as defined in Theorem~\ref{thm2}. For a sufficiently large $C'>0$, if 
\begin{equation}\label{Eq:thm3Assumption}\Bigl( \min_{1\leq \ell\leq N} \min\{\Delta^{\ell}_{p,n}\delta^{\ell-1}_{p,n}, \Delta^{\ell}_{p,n}\delta^{\ell}_{p,n}\} \Bigr) \geq C' pnR_{p,n},
\end{equation}
where $\Delta^\ell_{p, n} = \sum_{i \in \Omega_\ell} \big(f_{i, \eta_\ell+1} - f_{i, \eta_\ell} \big)^2$, for $1\leq \ell\leq N$ and $\delta_{p, n}^\ell = \eta_{\ell+1}-\eta_{\ell}$ for $0\leq \ell \leq N$, then we have 
\begin{equation} 
\label{Eq:Thm3Rate}
\mathbb{P} \; \bigg( \hat{N}=N, \quad \max_{\ell=1, \ldots, N} \Big\{ |\hat{\eta}_\ell-\eta_\ell| \cdot \Delta^\ell_{p, n} \Big\} \leq C pn R_{p, n}  \bigg) \; \rightarrow \; 1,
\end{equation}
as $n \rightarrow \infty$ where $C$ is a constant depending only on $C'$. 
\end{thm}

This theorem states that as long as each individual change-point is prominent enough in the sense that the signal size and its spacing away from its neighbors satisfy~\eqref{Eq:thm3Assumption}, then we can consistently estimate both the number and location of all the change-points. Moreover,~\eqref{Eq:thm3Assumption} describes the difficulty of estimating each individual change-point by its local energy statistic, as described in \citet{verzelen2020optimal}, and the rate of convergence for a given change-point in~\eqref{Eq:Thm3Rate} depends on its respective signal size $\Delta_{p,n}^{\ell}$.


\section{Numerical studies} \label{sec4}

\subsection{Parameter choice} \label{sec4.1}

\quad \;  \textbf{\emph{Choice of thresholds $\lambda_\infty$ and $\lambda_2$.}} 
As described in Section \ref{thresholding}, the BUHDA algorithm is built on the $L_\infty$ and $L_2$ aggregations of CUSUM statistics. The corresponding thresholds have the form of $\lambda_\infty = c_1\log^{1/2}(n)$ and $\lambda_2 = c_2\sqrt{p + \log n}$, where $c_1$ and $c_2$ are positive constants. The theoretical derivations of $\lambda_\infty$ and $\lambda_2$ can be found in Appendix \ref{pf}.
These thresholds are valid under the assumption that $\sigma_i=1$ for all $i=1, \ldots, p$, where $\varepsilon_{i,t} \sim N(0,\sigma_i^2)$. However, in practice $\sigma_i$ is often unknown and can vary across data sequences $X_i$. Thus we normalize each data sequence by its estimated standard deviation and use the thresholds described above. For this, we first estimate each $\sigma_i$ using the Median Absolute Deviation (MAD) estimator \citep{hampel1974influence}.
To choose the optimal thresholds $\lambda_2$ and $\lambda_\infty$ in practice, for given $(p, n)$, we first generate 100 datasets from the null model without change-points and compute the maximum $L_2$ and $L_\infty$ aggregated CUSUM statistics over $n$ randomly chosen intervals for each realization. We then take the $95\%$ quantiles of those maximum values as $\lambda_2$ and $\lambda_\infty$.

\quad \;  \textbf{\emph{Choice of $\rho$.}} 
$\rho \in (0, 1)$ is the parameter which decides the proportion of pairs of segments to merge in a single pass over the data. Empirically, the change detection performance is quite robust to the choice of this parameter. We use $\rho=0.3$ as a default in the simulation study and data analyses.

\quad \;  \textbf{\emph{Choice of the pre-merge pass parameter $c_{\mathrm{pm}}$.}} 
As shown in Algorithm \ref{algo_bu}, we set the number of pre-merge rounds to be $\sfloor{(\log_2{n})/ c_\text{pm}}$, where we set the default value of $c_\text{pm}$ to be 2, based on empirical simulations and the observed robustness of this parameter. 

\subsection{Simulation settings} \label{sec4.2}

\begin{table}[ht]
\centering
\resizebox{\textwidth}{!}{
\begin{tabular}{l|lcc}
  \hline
scenario &  sparsity & ($\mathcal{S}_1, \ldots, \mathcal{S}_N$) & ($\theta_1, \ldots, \theta_N$) \\
 \hline
\multirow{5}*{\parbox{4.5cm}{(LD): Low-Dimensional \\$n=120, p=50, N=3$ \\ $\b{\eta}=30 \cdot (1, 2, 3)$}} & sparse & (1, 1, 1) & (2.5, 2.5, 2.5)\\ 
& dense & (35, 35, 35) & (3, 3, 3)\\
& mixed1 & (1, 2, 7) & (2.8, 2.8, 2.8)\\
& mixed2 & (1, 1, 35) & (2.3, 2.3, 2.8)\\
& mixed3 & (1, 7, 35) & (2.3, 2.5, 2.8)\\
\cline{1-4}
\multirow{4}*{\parbox{4.5cm}{(HD): High-Dimensional \\$n=300, p=500, N=5$ \\ $\b{\eta}=50\cdot(1, 2, 3, 4, 5)$}} & sparse & (2, 2, 2, 2, 2) & (2.8, 2.8, 2.8, 2.8, 2.8)\\
& moderate & (22, 22, 22, 22, 22) & (3.8, 3.8, 3.8, 3.8, 3.8)\\
& dense & (350, 350, 350, 350, 350) & (4.3, 4.3, 4.3, 4.3, 4.3)\\
& mixed & (350, 2, 22, 22, 2) & (4, 2.5, 3.5, 3.5, 2.5)\\
\hline 
\end{tabular}}
\caption{Parameter settings with different types of sparsity for scenarios (LD) and (HD).}
\label{scenarios}
\end{table}

In the simulation studies, we focus on relatively frequent change-point settings; less frequent change-point settings can be found in Appendix \ref{extrasim}. We explore two sets of scenarios, low-dimensional (LD) settings when $n > p$ and high-dimensional (HD) settings when $n < p$, where the details can be found in Table \ref{scenarios}. For each set of scenarios, different types of sparsity are examined. \citet{liu2021minimax} study sparse change detection in a high-dimensional mean vector as a minimax testing problem and show that a phase transition occurs in the minimax testing rate when the sparsity level is of order $\sqrt{p \log\log(8n)}$. Following the settings used in \citet{zhang2022adaptive}, we call the sparsity level sparse if $\mathcal{S}_\ell = o(\sqrt{p})$, moderate if $\mathcal{S}_\ell \asymp \sqrt{p}$, dense if $\sqrt{p} = o(\mathcal{S}_\ell)$, and mixed if sparsity varies with change-points. Thus, under either set of scenarios in Table \ref{scenarios}, we respectively set $\mathcal{S}_\ell = \nint{0.1\cdot \sqrt{p}}$, $\mathcal{S}_\ell = \nint{\sqrt{p}}$ and $\mathcal{S}_\ell = \nint{0.7p}$. The last column in Table \ref{scenarios} shows the parameter for each signal strength, where $\theta_\ell=\sqrt{\Delta^{\ell}_{p,n}}$ describes the $\ell_2$ norm of change at each change-point $\eta_\ell$ and all $\mathcal{S}_\ell$ signal coordinates change by the same magnitude with the same sign.

\subsection{Competing methods}\label{sec4.3}
We perform the BUHDA procedure based on the parameter choice in Section \ref{sec4.1} and compare the performance with the following high-dimensional change-point detection methodologies: the sparsified binary segmentation ($\textbf{SBS}$, \citet{cho2015multiple}) and the double cusum algorithm ($\textbf{DC}$, \citet{cho2016change}) implemented in the R package \code{hdbinseg} and the informative sparse projection ($\textbf{INSPECT}$, \citet{wang2018high}) available in the R package \code{InspectChangepoint}, the adaptive self-normalization based approach ($\textbf{SN}$, \citet{zhang2022adaptive}), the scan statistic based algorithm ($\textbf{scanEH}$, \citet{enikeeva2019high}), the $L_\infty$ CUSUM aggregation algorithm ($\textbf{TD}(L_\infty)$, \citet{jirak2015uniform}), the $L_2$ CUSUM aggregation algorithm ($\textbf{TD}(L_2)$, \citet{horvath2012change}). The latter three methods are primarily designed for the testing problem for a single change-point, so we extend these methods to a multiple change-point estimation algorithm using the wild binary segmentation idea of \citet{fryzlewicz2014wild} with randomly chosen intervals. Thresholds for these three methods are chosen in the same way as for our algorithm, as described in Section \ref{sec4.1}. In implementing $\textbf{SN}$ \citep{zhang2022adaptive}, we use the adaptive WBS algorithm presented in Algorithm 2 of their paper that is designed for detecting multiple change-points. 
Whenever tuning parameters are required in running any of the methods introduced above, we follow the suggestions given by the authors in the relevant references. We also compare our BUHDA algorithm with variants of our method, where we only use $L_\infty$ CUSUM aggregation ($\textbf{BU}(L_\infty)$) or $L_2$ CUSUM aggregation ($\textbf{BU}(L_2)$).

\subsection{Simulation results}\label{sec4.4}

Under the simulation settings described in Section \ref{sec4.2}, we report two accuracy measures for the number of change-points: the empirical distribution of $\hat{N}-N$ and the mean squared errors (MSE) of the number of change-points detected over $100$ runs. As an accuracy measure for change-point locations, we report the average adjusted Rand index (ARI) of the estimated segmentation against the true one \citep{rand1971objective, hubert1985comparing}.

Table \ref{table:ld1} summarizes the results for sparse and dense settings considered in the low-dimensional scenarios (LD) where $n > p$. When only sparse changes exist, $L_\infty$ aggregation performs better than $L_2$ for both top-down and bottom-up methods. BUHDA shows comparable performance in terms of the number of change-points estimated. However, the location estimation accuracy, although good, is slightly inferior to some of the other competitors, e.g.\ SN and scanEH. In the dense case, the number of well-performing methods is not as high, and it is unexpected to see INSPECT and SBS performing similarly to the best method due to the dense nature of the change. Among the methods designed for adaptivity, BUHDA outperforms scanEH and remains not far behind SN in terms of correctly estimating all three changes. Remarkably, BUHDA achieves this good level of performance in much shorter computational time compared to SBS, DC, INSPECT, scanEH and SN (running time for SN is longer than 1 minute per repetition in MATLAB, though due to the difference in platform, we have not included its actual running time in the table).

In all three mixed sparsity settings presented in Table \ref{table:ld2}, $L_2$ aggregation performs better than $L_\infty$ aggregation for both top-down and bottom-up methods. While (LD)-mixed1 represents sparse to moderate change-points, the other two include dense change-points as well. With only a fraction of the running time compared to most competitors, BUHDA performs the best in terms of correctly identifying the number of change-points in (LD)-mixed2 and (LD)-mixed3, while showing comparable results to other adaptive methods in (LD)-mixed1. This shows that BUHDA can adapt to unknown sparsity levels of multiple change-points as it uses both norms in a data-adaptive way in building the bottom-up tree.

Table \ref{table:hd} shows the results of the high-dimensional scenario (HD) when $n < p$. 
As expected, TD($L_2$) and BU($L_2$) work well in the dense case, while TD($L_\infty$) and BU($L_\infty$) have the best performance in the sparse case. Excluding these four methods that directly use knowledge of sparsity, 
BUHDA shows competitive performance across sparse, moderate and dense settings, performing on par with the best remaining methods in each category (DC and scanEH in the sparse setting and INSPECT in the moderate and dense settings). More strikingly, in the mixed sparsity setting, BUHDA shows much better performance than all other methods, confirming its superior adaptivity to sparsity levels in the high-dimensional setting. As in the low-dimensional scenarios, BUHDA is computationally more efficient than the competing methods.


\begin{table}[htbp] 
\centering
\resizebox{\textwidth}{!}{
\begin{tabular}{crrrrrrrrrrr}
  \hline
  & & \multicolumn{7}{c}{$\hat{N}-N$} & & & \\
  \cline{3-9}
 Scenario & Method & $\leq$-3 & -2 & -1 & 0 & 1 & 2 & $\geq$3 & MSE & ARI & time\\ 
  \hline
\multirow{10}*{(LD)-sparse} 
  & TD($L_2$)  & 0 & 1 & 11 & 87 & 1 & 0 & 0 & 0.16 & 0.917 & 0.046 \\ 
  & TD($L_\infty$) & 0 & 0 & 0 & 100 & 0 & 0 & 0 & 0.00 & 0.970 & 0.057 \\ 
  & BU($L_2$) & 0 & 1 & 25 & 73 & 1 & 0 & 0 & 0.30 & 0.819 & 0.014 \\ 
  & BU($L_\infty$) & 0 & 0 & 4 & 94 & 2 & 0 & 0 & 0.06 & 0.876 & 0.015 \\ 
  & SBS & 0 & 51 & 45 & 4 & 0 & 0 & 0 & 2.49 & 0.505 & 0.050 \\ 
  & DC & 0 & 0 & 0 & 97 & 1 & 2 & 0 & 0.09 & 0.969 & 0.293 \\ 
  & INSPECT & 0 & 3 & 16 & 80 & 1 & 0 & 0 & 0.29 & 0.911 & 0.104 \\ 
  & scanEH & 0 & 0 & 0 & 99 & 0 & 1 & 0 & 0.04 & 0.969 & 0.283 \\ 
  & SN & 0 & 0 & 1 & 93 & 5 & 1 & 0 & 0.10 & 0.930 & \\ 
  & BUHDA & 0 & 0 & 5 & 93 & 2 & 0 & 0 & 0.07 & 0.874 & 0.019 \\ 
   \hline
  \multirow{10}*{(LD)-dense} 
  & TD($L_2$)  & 0 & 0 & 0 & 99 & 1 & 0 & 0 & 0.01 & 0.980 & 0.046 \\ 
  & TD($L_\infty$) & 0 & 66 & 32 & 2 & 0 & 0 & 0 & 2.96 & 0.466 & 0.056 \\ 
  & BU($L_2$) & 0 & 0 & 0 & 98 & 2 & 0 & 0 & 0.02 & 0.911 & 0.015 \\ 
  & BU($L_\infty$) & 0 & 87 & 12 & 1 & 0 & 0 & 0 & 3.60 & 0.429 & 0.013 \\ 
  & SBS & 0 & 0 & 2 & 97 & 1 & 0 & 0 & 0.03 & 0.938 & 0.052 \\ 
  & DC & 0 & 77 & 19 & 4 & 0 & 0 & 0 & 3.27 & 0.517 & 0.270 \\ 
  & INSPECT & 0 & 0 & 1 & 98 & 1 & 0 & 0 & 0.02 & 0.960 & 0.102 \\ 
  & scanEH & 0 & 1 & 20 & 78 & 0 & 1 & 0 & 0.28 & 0.903 & 0.277 \\ 
  & SN & 0 & 0 & 2 & 93 & 4 & 1 & 0 & 0.10 & 0.918 & \\ 
  & BUHDA & 0 & 0 & 3 & 87 & 10 & 0 & 0 & 0.13 & 0.854 & 0.017 \\ 
  \hline
\end{tabular}}
\caption {Distribution of $\hat{N}-N$ for scenarios (LD)-sparse and (LD)-dense and all methods listed in Section \ref{sec4.3} over 100 simulation runs. Also the average MSE (Mean Squared Error) of the number of change-points detected, the average Adjusted Rand index (ARI) of the estimated segmentation against the true one, the average computational time in seconds using an Intel Core i9 3.6 GHz CPU with 40 GB of RAM, all over 100 simulations.} \label{table:ld1}
\end{table}


\begin{table}[htbp] 
\centering
\resizebox{\textwidth}{!}{
\begin{tabular}{crrrrrrrrrrr}
  \hline
  & & \multicolumn{7}{c}{$\hat{N}-N$} & & & \\
  \cline{3-9}
 Scenario & Method & $\leq$-3 & -2 & -1 & 0 & 1 & 2 & $\geq$3 & MSE & ARI & time\\ 
  \hline
  \multirow{10}*{(LD)-mixed1} 
  & TD($L_2$) & 0 & 0 & 3 & 96 & 1 & 0 & 0 & 0.04 & 0.961 & 0.048 \\ 
  & TD($L_\infty$) & 0 & 0 & 36 & 59 & 5 & 0 & 0 & 0.41 & 0.833 & 0.064 \\ 
  & BU($L_2$) & 0 & 0 & 9 & 89 & 2 & 0 & 0 & 0.11 & 0.872 & 0.015 \\ 
  & BU($L_\infty$) & 0 & 0 & 62 & 38 & 0 & 0 & 0 & 0.62 & 0.739 & 0.016 \\ 
  & SBS & 0 & 1 & 68 & 29 & 1 & 1 & 0 & 0.77 & 0.668 & 0.062 \\ 
  & DC & 0 & 0 & 15 & 78 & 6 & 1 & 0 & 0.25 & 0.875 & 0.337 \\ 
  & INSPECT & 0 & 0 & 2 & 97 & 1 & 0 & 0 & 0.03 & 0.970 & 0.109 \\ 
  & scanEH & 0 & 0 & 0 & 99 & 0 & 1 & 0 & 0.04 & 0.974 & 0.302 \\ 
  & SN & 0 & 0 & 0 & 93 & 6 & 1 & 0 & 0.10 & 0.926 & \\ 
  & BUHDA & 0 & 0 & 2 & 93 & 5 & 0 & 0 & 0.07 & 0.878 & 0.018 \\ 
  \hline
  \multirow{10}*{(LD)-mixed2} 
  & TD($L_2$) & 0 & 1 & 15 & 83 & 1 & 0 & 0 & 0.20 & 0.907 & 0.048 \\ 
  & TD($L_\infty$) & 0 & 1 & 87 & 12 & 0 & 0 & 0 & 0.91 & 0.699 & 0.059 \\ 
  & BU($L_2$) & 0 & 0 & 36 & 63 & 1 & 0 & 0 & 0.37 & 0.803 & 0.015 \\ 
  & BU($L_\infty$) & 0 & 5 & 94 & 1 & 0 & 0 & 0 & 1.14 & 0.622 & 0.016 \\ 
  & SBS & 0 & 16 & 74 & 10 & 0 & 0 & 0 & 1.38 & 0.588 & 0.063 \\ 
  & DC & 0 & 1 & 85 & 11 & 3 & 0 & 0 & 0.92 & 0.689 & 0.311 \\ 
  & INSPECT & 0 & 1 & 25 & 73 & 1 & 0 & 0 & 0.30 & 0.872 & 0.111 \\ 
  & scanEH & 0 & 1 & 21 & 77 & 0 & 1 & 0 & 0.29 & 0.887 & 0.300 \\ 
  & SN & 0 & 0 & 3 & 90 & 6 & 1 & 0 & 0.13 & 0.923 & \\ 
  & BUHDA & 0 & 0 & 6 & 91 & 3 & 0 & 0 & 0.09 & 0.852 & 0.019 \\ 
  \hline
  \multirow{10}*{(LD)-mixed3} 
  & TD($L_2$) & 0 & 0 & 14 & 84 & 2 & 0 & 0 & 0.16 & 0.909 & 0.046 \\ 
  & TD($L_\infty$) & 0 & 1 & 92 & 7 & 0 & 0 & 0 & 0.96 & 0.658 & 0.062 \\ 
  & BU($L_2$) & 0 & 1 & 25 & 73 & 1 & 0 & 0 & 0.30 & 0.824 & 0.015 \\ 
  & BU($L_\infty$) & 0 & 14 & 85 & 1 & 0 & 0 & 0 & 1.41 & 0.593 & 0.014 \\ 
  & SBS & 0 & 1 & 85 & 13 & 1 & 0 & 0 & 0.90 & 0.674 & 0.060 \\ 
  & DC & 0 & 0 & 88 & 11 & 1 & 0 & 0 & 0.89 & 0.680 & 0.337 \\ 
  & INSPECT & 0 & 0 & 25 & 74 & 1 & 0 & 0 & 0.26 & 0.883 & 0.105 \\ 
  & scanEH & 0 & 0 & 23 & 76 & 1 & 0 & 0 & 0.24 & 0.886 & 0.309 \\ 
  & SN & 0 & 2 & 11 & 82 & 4 & 1 & 0 & 0.27 & 0.883 & \\ 
  & BUHDA & 0 & 0 & 12 & 85 & 3 & 0 & 0 & 0.15 & 0.830 & 0.017 \\  
   \hline
\end{tabular}}
\caption {Distribution of $\hat{N}-N$ for scenario (LD)-mixed1--3 and all methods listed in Section \ref{sec4.3} over 100 simulation runs. Also the average MSE (Mean Squared Error) of the number of change-points detected, the average Adjusted Rand index (ARI) of the estimated segmentation against the true one, the average computational time in seconds using an Intel Core i9 3.6 GHz CPU with 40 GB of RAM, all over 100 simulations.} \label{table:ld2}
\end{table}


\begin{table}[htbp] 
\centering
\footnotesize
\renewcommand{\arraystretch}{0.9}
\begin{tabular}{crrrrrrrrrrr}
  \hline
  & & \multicolumn{7}{c}{$\hat{N}-N$} & & &\\
  \cline{3-9}
 Scenario & Method & $\leq$-3 & -2 & -1 & 0 & 1 & 2 & $\geq$3 & MSE & ARI & time\\ 
  \hline
\multirow{10}*{(HD)-sparse} 
  & TD($L_2$) & 0 & 13 & 39 & 48 & 0 & 0 & 0 & 0.91 & 0.848 & 1.448 \\ 
  & TD($L_\infty$) & 0 & 0 & 0 & 86 & 13 & 1 & 0 & 0.17 & 0.953 & 1.568 \\ 
  & BU($L_2$) & 4 & 24 & 54 & 18 & 0 & 0 & 0 & 1.86 & 0.717 & 0.072 \\  
  & BU($L_\infty$) & 0 & 0 & 2 & 96 & 2 & 0 & 0 & 0.04 & 0.863 & 0.074 \\ 
  & SBS & 6 & 52 & 38 & 4 & 0 & 0 & 0 & 3.00 & 0.618 & 1.789 \\ 
  & DC & 0 & 0 & 0 & 93 & 7 & 0 & 0 & 0.07 & 0.969 & 4.859 \\ 
  & INSPECT & 0 & 41 & 36 & 23 & 0 & 0 & 0 & 2.00 & 0.788 & 2.772 \\ 
  & scanEH & 0 & 0 & 0 & 96 & 4 & 0 & 0 & 0.04 & 0.980 & 4.069 \\  
  & SN & 0 & 0 & 0 & 0 & 0 & 88 & 12 & 0.12 & 0.941 & \\ 
  & BUHDA & 0 & 0 & 3 & 94 & 3 & 0 & 0 & 0.06 & 0.858 & 0.084 \\ 
  \hline
\multirow{10}*{(HD)-moderate} 
  & TD($L_2$) & 0 & 0 & 0 & 99 & 1 & 0 & 0 & 0.01 & 0.980 & 1.401 \\ 
  & TD($L_\infty$) & 0 & 45 & 45 & 10 & 0 & 0 & 0 & 2.25 & 0.693 & 1.481 \\ 
  & BU($L_2$) & 0 & 0 & 1 & 98 & 1 & 0 & 0 & 0.02 & 0.880 & 0.069 \\ 
  & BU($L_\infty$) & 15 & 70 & 15 & 0 & 0 & 0 & 0 & 4.30 & 0.617 & 0.068 \\ 
  & SBS & 0 & 0 & 0 & 90 & 10 & 0 & 0 & 0.10 & 0.936 & 1.761 \\ 
  & DC & 0 & 20 & 41 & 39 & 0 & 0 & 0 & 1.21 & 0.832 & 4.734 \\ 
  & INSPECT & 0 & 0 & 0 & 98 & 2 & 0 & 0 & 0.02 & 0.986 & 2.750 \\  
  & scanEH & 0 & 0 & 0 & 99 & 1 & 0 & 0 & 0.01 & 0.986 & 4.008 \\ 
  & SN & 0 & 0 & 0 & 4 & 32 & 58 & 6 & 0.54 & 0.846 & \\ 
  & BUHDA & 0 & 0 & 7 & 92 & 1 & 0 & 0 & 0.08 & 0.853 & 0.080 \\ 
  \hline
\multirow{10}*{(HD)-dense} 
  & TD($L_2$) & 0 & 0 & 0 & 100 & 0 & 0 & 0 & 0.00 & 0.990 & 1.374 \\ 
  & TD($L_\infty$) & 99 & 1 & 0 & 0 & 0 & 0 & 0 & 11.61 & 0.389 & 1.440 \\ 
  & BU($L_2$) & 0 & 0 & 0 & 94 & 6 & 0 & 0 & 0.06 & 0.903 & 0.068 \\ 
  & BU($L_\infty$) & 100 & 0 & 0 & 0 & 0 & 0 & 0 & 14.46 & 0.325 & 0.067 \\  
  & SBS & 0 & 2 & 26 & 71 & 1 & 0 & 0 & 0.35 & 0.877 & 1.734 \\ 
  & DC & 0 & 100 & 0 & 0 & 0 & 0 & 0 & 4.00 & 0.673 & 4.695 \\ 
  & INSPECT & 0 & 0 & 0 & 79 & 19 & 2 & 0 & 0.27 & 0.914 & 2.739 \\ 
  & scanEH & 0 & 32 & 51 & 17 & 0 & 0 & 0 & 1.79 & 0.775 & 3.999 \\ 
  & SN & 0 & 0 & 0 & 0 & 7 & 83 & 10 & 0.17 & 0.911 & \\ 
  & BUHDA & 0 & 0 & 1 & 79 & 20 & 0 & 0 & 0.21 & 0.841 & 0.078 \\ 
  \hline
\multirow{10}*{(HD)-mixed} 
  & TD($L_2$) & 0 & 20 & 41 & 39 & 0 & 0 & 0 & 1.21 & 0.839 & 1.359 \\ 
  & TD($L_\infty$) & 0 & 57 & 39 & 4 & 0 & 0 & 0 & 2.67 & 0.677 & 1.513 \\ 
  & BU($L_2$) & 0 & 38 & 44 & 18 & 0 & 0 & 0 & 1.96 & 0.724 & 0.069 \\ 
  & BU($L_\infty$) & 0 & 99 & 1 & 0 & 0 & 0 & 0 & 3.97 & 0.617 & 0.069 \\  
  & SBS & 0 & 61 & 33 & 6 & 0 & 0 & 0 & 2.77 & 0.665 & 1.738 \\ 
  & DC & 0 & 46 & 51 & 3 & 0 & 0 & 0 & 2.35 & 0.710 & 4.856 \\ 
  & INSPECT & 0 & 10 & 45 & 45 & 0 & 0 & 0 & 0.85 & 0.852 & 2.697 \\ 
  & scanEH & 0 & 0 & 56 & 42 & 2 & 0 & 0 & 0.58 & 0.852 & 4.089 \\ 
  & SN & 0 & 0 & 0 & 3 & 52 & 41 & 4 & 0.68 & 0.797 & \\ 
  & BUHDA & 0 & 0 & 13 & 82 & 5 & 0 & 0 & 0.18 & 0.830 & 0.082 \\  
   \hline
\end{tabular}
\caption {Distribution of $\hat{N}-N$ for scenario (HD) and all methods listed in Section \ref{sec4.3} over 100 simulation runs. Also the average MSE (Mean Squared Error) of the number of change-points detected, the average Adjusted Rand index (ARI) of the estimated segmentation against the true one, the average computational time in seconds using an Intel Core i9 3.6 GHz CPU with 40 GB of RAM, all over 100 simulations.} \label{table:hd}
\end{table}

\clearpage

\subsection{Real data example} \label{sec4.5}

We analyze monthly percentage changes in the UK's average house price across 32 London boroughs, from January 1995 to June 2025. The average price is based on actual transaction data collected from HM Land Registry and includes all property types: detached, semi-detached and terraced houses, as well as flats and maisonettes. The data set is available from \url{https://landregistry.data.gov.uk/app/ukhpi/?lang=en}.

\begin{figure}[ht!] 
\centering
\includegraphics[width=15cm, height=11cm]{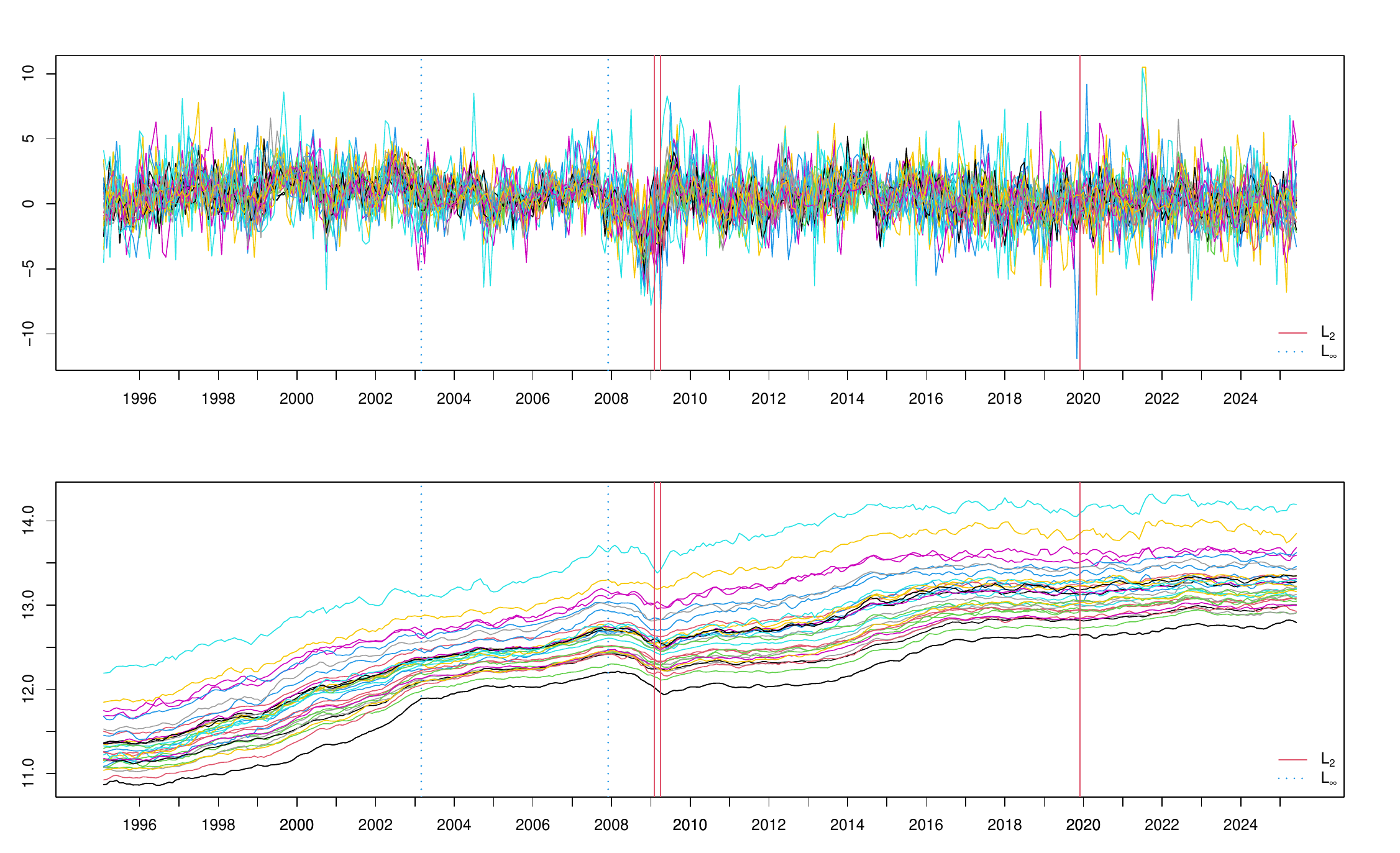}
 \caption{Top: the log-return dataset of London house prices and estimated change-points by BUHDA on this dataset, with colors denoting categories of changes (see Section~\ref{thresholding}); bottom: the same change-points overlaid on the log prices.}
\label{fig:houseprice}
\end{figure}

To address potential cross-sectional dependence, instead of the raw data, we used PCA-transformed data, namely the standardized matrix of principal component scores. Since the data matrix was already standardized, we skipped the normalization step prior to analysis. Applying our BUHDA algorithm using the threshold-selection method described in Section \ref{sec4.1} identified 5 change-points. The other methods, scanEH and INSPECT, detect 6 and 2 change-points, respectively. 

As BUHDA detects change-points via connected thresholding described in Section \ref{thresholding}, the detected change-points can be classified into three cases: surviving only by $\lambda_2$, only by $\lambda_\infty$, or by both $\lambda_2$ and $\lambda_\infty$. Figure \ref{fig:houseprice} shows that those surviving either $\lambda_2$ or $\lambda_\infty$ are located around known crises such as the global financial crisis (2007--2009) and the Covid-19 restrictions. As shown at the bottom of Figure \ref{fig:houseprice}, the change-points identified in the returns domain capture shifts in the log-price domain.

\ifnotblind{
\paragraph{Acknowledgments.}
} \fi

\bibliographystyle{apalike}
\bibliography{ref}

@article{morresi2024high,
  title={High-dimensional detection of Landscape Dynamics: a Landsat time series-based algorithm for forest disturbance mapping and beyond},
  author={Morresi, Donato and Maeng, Hyeyoung and Marzano, Raffaella and Lingua, Emanuele and Motta, Renzo and Garbarino, Matteo},
  journal={GIScience \& Remote Sensing},
  volume={61},
  pages={2365001},
  year={2024},
  publisher={Taylor \& Francis}
}

@article{wang2022inference,
  title={Inference for change points in high-dimensional data via selfnormalization},
  author={Wang, Runmin and Zhu, Changbo and Volgushev, Stanislav and Shao, Xiaofeng},
  journal={The Annals of Statistics},
  volume={50},
  pages={781--806},
  year={2022},
  publisher={Institute of Mathematical Statistics}
}

@article{chen2022inference,
  title={Inference of breakpoints in high-dimensional time series},
  author={Chen, Likai and Wang, Weining and Wu, Wei Biao},
  journal={Journal of the American Statistical Association},
  volume={117},
  pages={1951--1963},
  year={2022},
  publisher={Taylor \& Francis}
}

@article{tickle2021computationally,
  author  = {Tickle, S. O. and Eckley, I. A. and Fearnhead, P.},
  title   = {A Computationally Efficient, High-Dimensional Multiple Changepoint Procedure with Application to Global Terrorism Incidence},
  journal = {Journal of the Royal Statistical Society Series A: Statistics in Society},
  year    = {2021},
  volume  = {184},
  pages   = {1303--1325},
}

@book{bosq1998nonparametric,
  title={Nonparametric statistics for stochastic processes},
  author={Bosq, Denis},
  year={1998},
  publisher={Springer, New York}
}

@article{maeng2023detecting,
  title={Detecting linear trend changes in data sequences},
  author={Maeng, Hyeyoung and Fryzlewicz, Piotr},
  journal={Statistical Papers},
  volume={65},
  pages={1645--1675},
  year={2024},
  publisher={Springer}
}

@article{yu2021finite,
  title={Finite sample change point inference and identification for high-dimensional mean vectors},
  author={Yu, Mengjia and Chen, Xiaohui},
  journal={Journal of the Royal Statistical Society Series B: Statistical Methodology},
  volume={83},
  pages={247--270},
  year={2021},
  publisher={Oxford University Press}
}

@article{wang2023computationally,
  title={Computationally efficient and data-adaptive changepoint inference in high dimension},
  author={Wang, Guanghui and Feng, Long},
  journal={Journal of the Royal Statistical Society Series B: Statistical Methodology},
  volume={85},
  pages={936--958},
  year={2023},
  publisher={Oxford University Press US}
}

@article{liu2020unified,
  title={A unified data-adaptive framework for high dimensional change point detection},
  author={Liu, Bin and Zhou, Cheng and Zhang, Xinsheng and Liu, Yufeng},
  journal={Journal of the Royal Statistical Society Series B: Statistical Methodology},
  volume={82},
  pages={933--963},
  year={2020},
  publisher={Oxford University Press}
}

@article{rand1971objective,
  title={Objective criteria for the evaluation of clustering methods},
  author={Rand, William M},
  journal={Journal of the American Statistical Association},
  volume={66},
  pages={846--850},
  year={1971},
  publisher={Taylor \& Francis}
}

@article{hubert1985comparing,
  title={Comparing partitions},
  author={Hubert, Lawrence and Arabie, Phipps},
  journal={Journal of classification},
  volume={2},
  pages={193--218},
  year={1985},
  publisher={Springer}
}

@article{zhang2022adaptive,
  title={Adaptive inference for change points in high-dimensional data},
  author={Zhang, Yangfan and Wang, Runmin and Shao, Xiaofeng},
  journal={Journal of the American Statistical Association},
  volume={117},
  pages={1751--1762},
  year={2022},
  publisher={Taylor \& Francis}
}

@article{liu2021minimax,
  title={Minimax rates in sparse, high-dimensional change point detection},
  author={Liu, Haoyang and Gao, Chao and Samworth, Richard J},
  journal={The Annals of Statistics},
  volume={49},
  pages={1081--1112},
  year={2021},
}

@article{verzelen2020optimal,
  title={Optimal change-point detection and localization},
  author={Verzelen, Nicolas and Fromont, Magalie and Lerasle, Matthieu and Reynaud-Bouret, Patricia},
  journal={The Annals of Statistics},
  volume={51},
  pages={1586--1610},
  year={2023},
  publisher={Institute of Mathematical Statistics}
}

@article{laurent2000adaptive,
  title={Adaptive estimation of a quadratic functional by model selection},
  author={Laurent, Beatrice and Massart, Pascal},
  journal={The Annals of Statistics},
  volume={28},
  pages={1302--1338},
  year={2000},
  publisher={JSTOR}
}

@article{cribben2017estimating,
  title={Estimating whole-brain dynamics by using spectral clustering},
  author={Cribben, Ivor and Yu, Yi},
  journal={Journal of the Royal Statistical Society Series C: Applied Statistics},
  volume={66},
  pages={607--627},
  year={2017},
  publisher={Wiley Online Library}
}

@article{zhang2010detecting,
  title={Detecting simultaneous changepoints in multiple sequences},
  author={Zhang, Nancy R and Siegmund, David O and Ji, Hanlee and Li, Jun Z},
  journal={Biometrika},
  volume={97},
  pages={631--645},
  year={2010},
  publisher={Oxford University Press}
}

@article{bai2010common,
  title={Common breaks in means and variances for panel data},
  author={Bai, Jushan},
  journal={Journal of Econometrics},
  volume={157},
  pages={78--92},
  year={2010},
  publisher={Elsevier}
}

@article{li2019change,
  title={Change Point Detection in the Mean of High-Dimensional Time Series Data under Dependence},
  author={Li, Jun and Xu, Minya and Zhong, Ping-Shou and Li, Lingjun},
  journal={arXiv preprint arXiv:1903.07006},
  year={2019}
}

@article{groen2013multivariate,
  title={Multivariate methods for monitoring structural change},
  author={Groen, Jan JJ and Kapetanios, George and Price, Simon},
  journal={Journal of Applied Econometrics},
  volume={28},
  pages={250--274},
  year={2013},
  publisher={Wiley Online Library}
}

@article{horvath2012change,
  title={Change-point detection in panel data},
  author={Horv{\'a}th, Lajos and Hu{\v{s}}kov{\'a}, Marie},
  journal={Journal of Time Series Analysis},
  volume={33},
  pages={631--648},
  year={2012},
  publisher={Wiley Online Library}
}

@article{enikeeva2019high,
  title={High-dimensional change-point detection with sparse alternatives},
  author={Enikeeva, Farida and Harchaoui, Zaid},
  journal={The Annals of Statistics},
  volume={47},
  pages={2051--2079},
  year={2019}
}

@article{bardwell2019most,
  title={Most recent changepoint detection in panel data},
  author={Bardwell, Lawrence and Fearnhead, Paul and Eckley, Idris A and Smith, Simon and Spott, Martin},
  journal={Technometrics},
  volume={61},
  pages={88--98},
  year={2019},
  publisher={Taylor \& Francis}
}

@article{jirak2015uniform,
  title={Uniform change point tests in high dimension},
  author={Jirak, Moritz},
  journal={The Annals of Statistics},
  volume={43},
  pages={2451--2483},
  year={2015},
  publisher={Institute of Mathematical Statistics}
}

@article{wang2018high,
  title={High dimensional change point estimation via sparse projection},
  author={Wang, Tengyao and Samworth, Richard J},
  journal={Journal of the Royal Statistical Society Series B: Statistical Methodology},
  volume={80},
  pages={57--83},
  year={2018},
  publisher={Wiley Online Library}
}

@article{cho2015multiple,
  title={Multiple-change-point detection for high dimensional time series via sparsified binary segmentation},
  author={Cho, Haeran and Fryzlewicz, Piotr},
  journal={Journal of the Royal Statistical Society Series B: Statistical Methodology},
  volume={77},
  pages={475--507},
  year={2015},
  publisher={Wiley Online Library}
}

@article{cho2016change,
  title={Change-point detection in panel data via double CUSUM statistic},
  author={Cho, Haeran and others},
  journal={Electronic Journal of Statistics},
  volume={10},
  pages={2000--2038},
  year={2016},
  publisher={The Institute of Mathematical Statistics and the Bernoulli Society}
}

@article{hampel1974influence,
  title={The influence curve and its role in robust estimation},
  author={Hampel, Frank R},
  journal={Journal of the american statistical association},
  volume={69},
  pages={383--393},
  year={1974},
  publisher={Taylor \& Francis}
}

@article{fisch2022linear,
  title={A linear time method for the detection of collective and point anomalies},
  author={Fisch, Alexander TM and Eckley, Idris A and Fearnhead, Paul},
  journal={Statistical Analysis and Data Mining: The ASA Data Science Journal},
  volume={15},
  pages={494--508},
  year={2022},
  publisher={Wiley Online Library}
}

@article{fryzlewicz2017tail,
  title={Tail-greedy bottom-up data decompositions and fast mulitple change-point detection},
  author={Fryzlewicz, Piotr},
  journal={The Annals of Statistics},
  volume={46},
  pages={3390--3421},
  year={2018},
  publisher={Institute of Mathematical Statistics}
}

@article{fryzlewicz2014wild,
  title={Wild binary segmentation for multiple change-point detection},
  author={Fryzlewicz, Piotr},
  journal={The Annals of Statistics},
  volume={42},
  pages={2243--2281},
  year={2014},
  publisher={Institute of Mathematical Statistics}
}

\clearpage

\appendix

\numberwithin{equation}{section}
\numberwithin{figure}{section}
\numberwithin{table}{section}

\section{Proofs} \label{pf}
The proofs of Theorems \ref{thm1}-\ref{thm3} are given below. 

\subsection{Some useful lemma}\label{hits.lemmas}
We first present a preparatory lemma.



\begin{lem} \label{lm5.1}
Suppose that $\{\b{X}_{i, \cdot}\}_{i=1}^p$ follow model \eqref{model} with $\sigma_i=1$ for all $i=1, \ldots, p$. Assume $\lambda_\infty = c_1\log^{1/2} (n)$ for a constant $c_1 \geq \sqrt{2(3+\alpha)}$ and the dimension $p$ satisfies $p \lesssim n^\alpha$ for some fixed $\alpha \in (0, \infty)$, then we have $P(A_{p, n}) \rightarrow 1$ as $n \rightarrow \infty$ where 
\begin{equation} \label{e.ap.5.1}
A_{p, n}=\bigg\{ \forall 0 \leq u < v < w \leq n \quad \max_{i} \; |\langle \psi_{u, v, w}, \bve_i \rangle | \leq \lambda_{\infty} \bigg\},
\end{equation} 
and $\psi_{u, v, w}$ is a $n$-dimensional vector, where each component $\psi_{u, v, w, t}$ is 
\begin{equation}
  \psi_{u, v, w, t} =
    \begin{cases}
      \sqrt{\frac{w-v}{(v-u)(w-u)}} & \text{if $u<t\leq v$},\\
      -\sqrt{\frac{v-u}{(w-v)(w-u)}} & \text{if $v<t\leq w$},\\
      0 & \text{otherwise}.
    \end{cases}       
\end{equation}

\end{lem}

\textbf{Proof.} 
Using a simple Bonferroni inequality, we have
\begin{align*} 
1-P(A_{p, n}) & \leq \sum_{i=1}^p \sum_{(u, v, w)} P(|Z|>\lambda_{\infty}) \\
& \leq 2pn^3 \frac{\phi_Z(\lambda_{\infty})}{\lambda_{\infty}} \\
& = 2pn^3 \frac{e^{-\lambda_{\infty}^2/2} / \sqrt{2\pi}}{\lambda_{\infty}} \\ \numberthis \label{a11}
& = \frac{2}{\sqrt{2\pi}}\cdot \frac{1}{c_1\sqrt{\log(n)}} \cdot \frac{n^{(3+\alpha)}}{n^{c_1^2/2}} \\
& \rightarrow 0,
\end{align*}
as $n \rightarrow \infty$, where $\phi_Z$ is the p.d.f. of a standard normal $Z$. The second inequality in \eqref{a11} holds because
\begin{equation} \label{a12}
P(|Z|>\lambda_{\infty}) = 2\frac{1}{\sqrt{2\pi}} \int_{\lambda_{\infty}}^\infty e^{-x^2/2} dx \leq 2\frac{1}{\sqrt{2\pi}} \int_{\lambda_{\infty}}^\infty \frac{x}{\lambda_{\infty}} e^{-x^2/2} dx = 2 \frac{\phi_Z(\lambda_{\infty})}{\lambda_{\infty}},
\end{equation}
and the remaining parts because of the definition of $\phi_Z$ and $\lambda_\infty$ and the condition given on $c_1$. This completes the proof.


\vspace{10pt}


\begin{lem} \label{lm5.2}
Suppose that $\{\b{X}_{i, \cdot}\}_{i=1}^p$ follow model \eqref{model} with $\sigma_i=1$ for all $i=1, \ldots, p$. If we set the threshold $\lambda_{2} = c_2 \sqrt{p + \log n}$ with a sufficiently large constant $c_2 > 0$ and the dimension satisfies $p\lesssim n^\alpha$ for some fixed $\alpha \in (0, \infty)$, then we have $P(B_{p, n}) \rightarrow 1$ as $n \rightarrow \infty$ where 
\begin{equation} \label{e.ap.5.2}
B_{p, n}=\bigg\{ \forall 0 \leq u < v < w \leq n \quad \bigg\{ \sum_{i=1}^p \; |\langle \psi_{u, v, w}, \bve_i \rangle |^2 \bigg\}^{1/2}  \leq \lambda_{2} \bigg\}. 
\end{equation}
\end{lem}

\textbf{Proof.} 
Using a simple Bonferroni inequality, we have
\begin{equation} 
1-P(B_{p, n}) \leq  \sum_{(u, v, w)} P(\lVert Z \rVert_2>\lambda_{2}) = \sum_{(u, v, w)} P(\lVert Z \rVert^2_2>\lambda_{2}^2) \leq n^3 \, e^{-C\log n} = n^{3-C}
\end{equation}
where $Z \sim N(\boldsymbol{0}, \textbf{I}_p)$. The last inequality holds by Lemma 1 of \citet{laurent2000adaptive}, which states that $P(\lVert Z \rVert_2^2> p + 2\sqrt{px} + 2x) \leq e^{-x}$. Taking $x = C\log n$ and using $2\sqrt{px} \leq p + x$, we have $p + 2\sqrt{px} + 2x \leq 2p + 3x \leq c_2^2(p + \log n) = \lambda_2^2$ for a sufficiently large constant $c_2$, so that $P(\lVert Z \rVert_2^2 > \lambda_2^2) \leq e^{-C\log n}$. Choosing $C > 3$ gives $n^{3-C} \to 0$ as $n \to \infty$. This completes the proof.


\vspace{10pt}

\subsection{Proof of Theorems \ref{thm1} - \ref{thm3}}\label{hits.pfthm1}

\noindent {\bf Proof of Theorem \ref{thm1}.}

From the thresholding in Section \ref{thresholding}, the estimator of ${\mu}_{i; u, v, w} = \langle \b f_{i, \cdot}, \psi_{u, v, w} \rangle$ for $i = 1, \ldots, p$ can be obtained as
\begin{equation} \label{muhat0}
\begin{split}
\hat{\mu}_{i; u, v, w} = C_{i; u, v, w} \cdot \mathbb{I} \; \Big\{ \, & \exists \text{ a child segment $(u', w']$ of $(u, w]$} \\
& \text{ such that } C^{L_\infty}_{u' ,v' ,w'} > \lambda_\infty \text{ or } C^{L_2}_{u' ,v' ,w'} > \lambda_2  \, \Big\}, 
\end{split}
\end{equation}
where $\mathbb{I}$ is the indicator function. Here we introduce a different notation of $\hat{\mu}_{i, [u, v, w]}$ in \eqref{muhat0} as follows:
\begin{equation} \label{muhat}
\begin{split}
\hat{\mu}_i^{(j, k)} = C^{(j, k)}_{i} \cdot \mathbb{I} \; \Bigg\{ \, & \exists (j', k') \in \mathcal{O}_{j, k} \quad \max_i\big|C_i^{(j', k')}\big| > \lambda_\infty \\
& \text{ or } \bigg\{\sum_{i=1}^p \big(C_i^{(j', k')}\big)^2 \bigg\}^{1/2} > \lambda_2 \, \Bigg\}, \; i=1,\ldots, p,
\end{split}
\end{equation}
where 
\begin{equation} \label{e5.2.11}
\begin{split}
\mathcal{O}_{j, k} = \{(j', k'), \, & j' = 1, \ldots, j, \; k' = 1, \ldots, K(j'): \\
& C^{(j', k')}_{\cdot; u' ,v' ,w'}  \text{ is such that } (u', w']  \text{ is a child segment of } (u, w] \}.
\end{split}
\end{equation}
Note that $\hat{\mu}_{i; u, v, w}$ and $\hat{\mu}_i^{(j, k)}$ have one-to-one correspondence and the set $\mathcal{O}_{j, k}$ includes all pairs of $(j', k')$ whose corresponding CUSUM statistic $C^{(j', k')}_{\cdot; u' ,v' ,w'}$ is a child of $C^{(j, k)}_{\cdot; u ,v ,w}$.

Let $\mathcal{R}^1 = \big\{(j, k),\; j = 1, \ldots, J, \; k = 1, \ldots, K(j)  : C_{\cdot; u ,v ,w}^{(j, k)}$ is such that $u < \eta_\ell + 1/2 < w$ for some $\ell=1, \ldots, N \big\}$, and $\mathcal{R}^0 = \{(j, k),\; j = 1, \ldots, J, \; k = 1, \ldots, K(j) \} \setminus \mathcal{R}^1$. Due to the orthonormality of the wavelet transform, on the event $A_{p, n} \cap B_{p, n}$, which has probability $1-o(1)$, defined in Lemmas \ref{lm5.1} and \ref{lm5.2}, we have
{\small 
\begin{align*} 
& \| {\tilde{\b f}}-\b{f} \|_{p, n}^2 \;  \\
& = \;  \frac{1}{p} \frac{1}{n} \sum_{i=1}^p  \sum_{j=1}^J \sum_{k=1}^{K(j)} \bigg(C_i^{(j, k)} \cdot \mathbb{I} \Big\{ \, \exists (j', k') \in \mathcal{O}_{j, k} \quad C^{L_\infty(j', k')} > \lambda_\infty \text{ or } C^{L_2(j', k')} > \lambda_2 \, \Big\} - \mu_i^{(j, k)} \bigg)^2 \\
& \quad  + \; \frac{1}{pn} \sum_{i=1}^p (s_{i, [1, n]} - \mu_i^{(0, 1)})^2 \\
& \; \leq \; \frac{1}{p} \frac{1}{n} \sum_{i=1}^p \bigg( \sum_{(j, k) \in \mathcal{R}^0} + \sum_{(j, k) \in \mathcal{R}^1} \Bigg) \\
& \qquad \times \Bigg( C_i^{(j, k)} \cdot \mathbb{I} \Big\{ \, \exists (j', k') \in \mathcal{O}_{j, k} \quad C^{L_\infty(j', k')} > \lambda_\infty \text{ or } C^{L_2(j', k')} > \lambda_2 \, \Big\} - \mu_i^{(j, k)} \bigg)^2 \\
& \quad + \; c_1^2n^{-1} \log (n) \\
& \; =: \; \mathit{I} + \mathit{II} + c_1^2n^{-1} \log (n) \numberthis ,
\end{align*}
}
where 
\begin{align*}
   &C^{L_\infty(j', k')} = \max_i\big|C_i^{(j', k')}\big|, \\
   &C^{L_2(j', k')} = \bigg\{\sum_{i=1}^p \big(C_i^{(j', k')}\big)^2 \bigg\}^{1/2}
\end{align*}

Since $(j, k) \in \mathcal{R}^0$, on the set $A_{p, n} \cap B_{p, n}$, by Lemmas \ref{lm5.1} and \ref{lm5.2}, we have $\max_i\big|C_i^{(j', k')}\big| \leq \lambda_\infty$ and $\big\{\sum_{i=1}^p \big(C_i^{(j', k')}\big)^2 \big\}^{1/2} \leq \lambda_2$ for all $(j', k') \in \mathcal{O}_{j, k}$. Also, by the fact that $\mu_i^{(j, k)} = 0$ for $(j,k) \in \mathcal{R}^0$ and $i=1, \ldots, p$, we have $\mathit{I} = 0$.

For $\mathit{II}$, we denote 
\begin{align*}
\mathcal{B} =  \bigg\{ \, \exists (j', k') \in \mathcal{O}_{j, k} \; \max_i \big|C_i^{(j', k')} \big| > \lambda_\infty \text{ or } \bigg\{\sum_{i=1}^p \big(C_i^{(j', k')}\big)^2 \bigg\}^{1/2} > \lambda_2 \, \bigg\}
\end{align*}
and have
\begin{align*} 
&  \frac{1}{pn} \sum_{i=1}^p \sum_{(j, k) \in \mathcal{R}^1}\Big( C_i^{(j, k)} \cdot \mathbb{I} \big\{ \mathcal{B} \big\} - \mu_i^{(j, k)} \Big)^2  \\
& =  \frac{1}{pn} \sum_{i=1}^p \sum_{(j, k) \in \mathcal{R}^1}\Big( C_i^{(j, k)} \cdot \mathbb{I} \big\{ \mathcal{B} \big\} - C_i^{(j, k)} + C_i^{(j, k)} - \mu_i^{(j, k)} \Big)^2  \\
&  \leq  \frac{1}{pn} \sum_{(j, k) \in \mathcal{R}^1} \sum_{i=1}^p \Bigg[ 2\Big( C_i^{(j, k)} \Big)^2 \cdot \mathbb{I}\bigg( \max_i \big|C_i^{(j, k)} \big| \leq \lambda_\infty \text{ and }\bigg\{\sum_{i=1}^p \big(C_i^{(j', k')}\big)^2 \bigg\}^{1/2} \leq \lambda_2 \bigg) \\
& \qquad\qquad\qquad\quad + 2\Big( C_i^{(j, k)} - \mu_i^{(j, k)}  \Big)^2 \Bigg]\\
&  \leq  \frac{1}{pn} \sum_{(j, k) \in \mathcal{R}^1}  \max_\ell \Big\{\mathcal{S}_\ell \cdot 4 c_1^2 \log (n) \wedge 4 c_2^2(p + \log n) \Big\}. \numberthis \label{ub} 
\end{align*}
The last inequality is obtained from the facts that
\begin{align*}
&  2\Big( C_i^{(j, k)} \Big)^2 \cdot \mathbb{I}\Big( \max_i \big|C_i^{(j, k)} \big| \leq \lambda_\infty \Big) +  2\Big( C_i^{(j, k)} - \mu_i^{(j, k)}  \Big)^2 \leq 2\lambda_\infty^2 + 2 c_1^2 \log (n) 
\end{align*}
and
\begin{align*} 
& 2\Bigg\{ \sum_{i=1}^p \bigg( C_i^{(j, k)} \cdot \mathbb{I}\Big( \Big\{ \sum_{i=1}^p (C_i^{(j', k')})^2 \Big\}^{1/2} \leq \lambda_2 \Big) \bigg)^2 +  \sum_{i=1}^p \Big( C_i^{(j, k)} - \mu_i^{(j, k)}  \Big)^2  \Bigg\} \leq 2\lambda_2^2 + 2c_2^2p 
\end{align*}

Combining with the upper bound of $J$, we have $|\mathcal{R}^1| \leq N \lceil \log (n) / \log (1/(1-\rho))\rceil$. Therefore, from \eqref{ub} we have 
\begin{equation} \label{l2_infty}
\| {\tilde{\b f}}-\b{f} \|_{p, n}^2 \; \leq \; \frac{1}{n} \Bigg[\min\big(c_1^2 \log(n), c_2^2\frac{p+\log n}{p}\big) + 4N \biggl\lceil\frac{\log (n)}{\log (1/(1-\rho))}\biggr \rceil \max_\ell \bigg\{ c_1^2 \frac{\mathcal{S}_\ell\log (n)}{p} \wedge c_2^2\frac{p+\log n}{p} \bigg\} \Bigg]
\end{equation}
Also, at each scale, the estimated change-points are obtained up to size $N$, combining it with the largest scale $J$, the number of change-points in ${\tilde{\b f}}$ is up to $O(N\log(n))$.


\vspace{10pt}

\noindent {\bf Proof of Theorem \ref{thm2}.}
The estimator ${\dbtilde{\b f}}$ is obtained by repeating the first three steps of the BUHDA procedure but applying only merge passes in tree construction in a greedy way, which allows only one merge at each scale. Thus the change-points in ${\dbtilde{\b f}}$ are a subset of those in ${\tilde{\b f}}$. Let $\tilde{B}$ and $\dbtilde{B}$ be the wavelet bases corresponding to ${\tilde{\b f}}$ and ${\dbtilde{\b f}}$, respectively. Then $\dbtilde{B}$ is classified into two categories: 1) all basis vectors $\psi^{(j, k)} \in \tilde{B}$ such that $\psi^{(j, k)}$ is not associated with the change-points in ${\tilde{\b f}}$ and the conditions, $\max_i \big| C_i^{(j, k)} \big|  < \lambda_\infty$ and $\big\{\sum_{i=1}^p \big(C_i^{(j', k')}\big)^2 \big\}^{1/2} < \lambda_2$, are satisfied and 2) all basis vectors $\tilde{\psi}^{(j, 1)}$ produced in Stage 1 of post-processing.

As the number of scales used plays an important role in controlling the $l_2$ behavior of $\| {\tilde{\b f}}-\b{f} \|_{p, n}^2$, we now investigate how many scales are used in each category defined above. In the first category, $C^{(j, k)}$ corresponding to the basis vectors $\psi^{(j, k)} \in \tilde{B}$ live on no more than $J$ scales and thus we have $|\mathcal{R}^1| \leq N \lceil \log (n) / \log (1/(1-\rho))\rceil$ by the same argument used in the proof of Theorem \ref{thm1}. Considering the second category, the basis vectors $\psi^{(j, 1)}$ correspond to different change-points in $\tilde{f}$ and there exist at most $\tilde{N}=O(N\log n)$ change-points in $\tilde{\b f}$  which we examine one at once, thus at most $\tilde{N}$ scales are required for $C^{(j, 1)}$. 
Combining these results with \eqref{ub}, the equivalent quantity of $\mathit{II}$ for $\dbtilde{\b f}$ is kept the same as $\tilde{\b f}$, and this gives us the following $l_2$ result
\begin{equation}
    \big\| {\dbtilde{\b f}} - \b{f} \big\|_{p, n}^2 \; =  \; O\Bigg( N\frac{\log n}{n} \max_\ell \bigg\{ \frac{\mathcal{S}_\ell\log n}{p} \wedge \Big(1+\frac{\log n}{p}\Big) \bigg\} \Bigg)
\end{equation}
under the same assumptions given in Theorem \ref{thm1}.

Finally, we show that there exist at most two change-points in $\dbtilde{\b f}$ between two consecutive change-points $(\eta_\ell, \eta_{\ell+1})$ for $\ell=0, \ldots, N$, where $\eta_0=0$ and $\eta_{N+1}=n$. If three estimated change-points, ($\dbtilde{\eta}_l, \dbtilde{\eta}_{l+1}, \dbtilde{\eta}_{l+2}$), lie between a pair of true change-points, $(\eta_\ell, \eta_{\ell+1})$, then $\max_i \Big|C_{i; \dbtilde{\eta}_l, \dbtilde{\eta}_{l+1}, \dbtilde{\eta}_{l+2}}\Big| < \lambda_\infty$ and $\big\{\sum_{i=1}^p \big(C_{i; \dbtilde{\eta}_l, \dbtilde{\eta}_{l+1}, \dbtilde{\eta}_{l+2}}\big)^2 \big\}^{1/2} < \lambda_2$ by Lemmas \ref{lm5.1} and \ref{lm5.2}. In other words, both $L_\infty$- and $L_2$-aggregated CUSUM statistics computed from the adjacent intervals, $(\dbtilde{\eta}_l, \dbtilde{\eta}_{l+1}]$ and $(\dbtilde{\eta}_{l+1}, \dbtilde{\eta}_{l+2}]$, are less than $\lambda_\infty$ and $\lambda_2$, respectively, so $ \dbtilde{\eta}_{l+1}$ would be removed from the set of estimated change-points. This satisfies $\dbtilde{N} \leq 2(N+1)$.


\vspace{10pt}

\noindent {\bf Proof of Theorem \ref{thm3}.}
We work on the event, whose probability approaches 1, on which the conclusion of Theorem~\ref{thm2} holds. In particular, we assume that for some sufficiently large $C>0$, we have $\big\| {\dbtilde{\b f}} - \b{f} \big\|_{p, n}^2 \leq \frac{C}{4} R_{p, n}$. Furthermore, from~\eqref{Eq:thm3Assumption}, we have for some large $C'>0$ that $C'pnR_{p, n} \big(\Delta^\ell_{p, n}\big)^{-1}  < \min\{\delta_{p, n}^{\ell-1},\delta_{p, n}^\ell\}$ for all $\ell=1, \ldots, N$. We write $r^\ell_{p, n}:=\lfloor CpnR_{p, n} \big(\Delta^\ell_{p, n}\big)^{-1}  \rfloor$ for $\ell \in \{1,\ldots,N\}$.

Suppose for at least one $\ell \in \{1,\ldots, N\}$, there is no estimated change-point in ${\dbtilde{\b f}}$ within distance of $r^\ell_{p, n}$ of $\eta_\ell$. This implies that  ${\dbtilde{f}_{i, j}}$ is constant over the entire segment $\{\eta_\ell-r^\ell_{p, n}, \ldots, \eta_\ell+r^\ell_{p, n}\}$ for all $i=1, \ldots, p$. Hence 
\begin{equation} \label{A14}
\big\| {\dbtilde{\b f}} - \b{f} \big\|_{p, n}^2 \geq \frac{1}{pn} \sum_{i \in \Omega_\ell} \sum_{j=\eta_\ell-r^\ell_{p, n}}^{\eta_\ell+r^\ell_{p, n}} (\dbtilde{f}_{i, j} - f_{i, j})^2  \geq \frac{r^\ell_{p, n}}{2pn} \sum_{i \in \Omega_\ell} \big(f_{i, \eta_\ell+1} - f_{i, \eta_\ell} \big)^2  > \frac{C}{4} R_{p, n},
\end{equation} 

From here the same argument used in the proof of Theorem 3.3 in \citet{fryzlewicz2017tail} is applied. Throughout Stage 2 of our post-processing in Section \ref{pp2}, there exist two possible scenarios for each index $\ell_0$ of an estimated change-point; $\dbtilde{\eta}_{\ell_0}$ is either the closest estimated change-point of any $\eta_\ell$ or not.
\begin{enumerate}
    \item Suppose that $\dbtilde{\eta}_{\ell_0}$ is not the closest estimated change-point to the nearest true change-point. Without loss of generality, assume that $\eta_\ell$ is the nearest true change-point to $\dbtilde{\eta}_{\ell_0}$ and that $\eta_\ell \geq \dbtilde{\eta}_{\ell_0}$ (the case $\eta_\ell \leq \dbtilde{\eta}_{\ell_0}$ can be handled symmetrically). Since there exists at least one estimated change-point within the distance of $r_{p,n}^{\ell-1}$ to the true change-point $\eta_{\ell-1}$, we must have $\dbtilde{\eta}_{\ell_0-1} \geq \eta_{\ell-1} - r_{p,n}^{\ell-1}$ and $\dbtilde{\eta}_{\ell_0} \geq \eta_{\ell-1} + (\eta_{\ell}-\eta_{\ell-1})/2$. By condition~\eqref{Eq:thm3Assumption}, for a sufficiently large $C'$, we have $(\eta_{\ell} - \eta_{\ell-1})/2 = \delta^{\ell-1}_{p,n} / 2 \geq \frac{C'pn R_{p,n}}{2\Delta_{p,n}^{\ell-1}} \geq r_{p,n}^{\ell-1}$. Consequently, $u_{\ell_0} = \lfloor (\dbtilde{\eta}_{\ell_0-1} + \dbtilde{\eta}_{\ell_0} )/2\rfloor$ as defined in Stage 2 of the refine satisfies that $u_{\ell_0} \geq \eta_{\ell-1}$. On the other hand, since $\dbtilde{\eta}_{\ell_0}$ is not the closest estimated change for $\eta_{\ell}$, we must have $w_{\ell_0} = \lfloor (\dbtilde{\eta}_{\ell_0} + \dbtilde{\eta}_{\ell_0+1} )/2\rfloor < \eta_\ell$. Therefore, the interval $(u_{\ell_0}, w_{\ell_0}]$ contains no true change-point and Lemmas \ref{lm5.1} and \ref{lm5.2} guarantee that the corresponding CUSUM statistic satisfies $C^{L_2}_{u_{\ell_0}, v_{\ell_0}, w_{\ell_0}} \leq \lambda_{2} \text{ and } C^{L_\infty}_{u_{\ell_0}, v_{\ell_0}, w_{\ell_0}} \leq \lambda_{\infty}$ and thus $\dbtilde{\eta}_{\ell_0}$ gets removed.
    \item Suppose $\dbtilde{\eta}_{\ell_0}$ is the closest estimated change-point of a true one, $\eta_\ell$, and therefore it is within the distance of $r^\ell_{p, n}$ from $\eta_\ell$. If $\dbtilde{\eta}_{\ell_0}$ gets removed as it so happens that $C^{L_2}_{u_{\ell_0}, v_{\ell_0}, w_{\ell_0}} \leq \lambda_{2}$ and $C^{L_\infty}_{u_{\ell_0}, v_{\ell_0}, w_{\ell_0}} \leq \lambda_{\infty}$, there must be another $\dbtilde{\eta}_{j}$ within the distance of $C_{\ell_0}\cdot r^\ell_{p, n}$ from $\dbtilde{\eta}_{\ell_0}$ (and thus from $\eta_\ell$), where $C_{\ell_0}$ is a constant. 
    If there were no such $\dbtilde{\eta}_{j}$ on either side of $\dbtilde{\eta}_{\ell_0}$, as the segment $(u_{\ell_0}, v_{\ell_0}]$ includes the true change-point $\eta_\ell$, either $C^{L_2}_{u_{\ell_0}, v_{\ell_0}, w_{\ell_0}} > \lambda_{2}$ or $C^{L_\infty}_{u_{\ell_0}, v_{\ell_0}, w_{\ell_0}} > \lambda_{\infty}$ would occur by the construction and $\dbtilde{\eta}_{\ell_0}$ would not get removed. More in detail, depending on the sparsity level, the following holds.
    \begin{align}
C^{L_\infty}_{u_{\ell_0}, v_{\ell_0}, w_{\ell_0}} \geq \sqrt{\frac{r_{p, n}^\ell \Delta_{p, n}^\ell}{2\mathcal{S}}} \geq \sqrt{\frac{C_1pnR_{p, n}}{2\mathcal{S}}} \asymp \log^{3/2} n  > \lambda_\infty, \quad \text{if } \; \mathcal{S} \leq \frac{p}{\log n},\label{linfty_alt} \\ 
C^{L_2}_{u_{\ell_0}, v_{\ell_0}, w_{\ell_0}} \geq \sqrt{\frac{r_{p, n}^\ell \Delta_{p, n}^\ell}{2}} \geq \sqrt{\frac{C_2pnR_{p, n}}{2}} \asymp \sqrt{p}\log n > \lambda_2, \quad \text{if } \; \mathcal{S} > \frac{p}{\log n}, \label{l2_alt}
    \end{align}
where $\mathcal{S} = \max_\ell \mathcal{S}_\ell$. The first inequalities in \eqref{linfty_alt} and \eqref{l2_alt} hold by the definition of CUSUM statistics and the second hold because we already set $r^\ell_{p, n}:=\lfloor CpnR_{p, n} \big(\Delta^\ell_{p, n}\big)^{-1} \rfloor$ for $\ell \in \{1,\ldots,N\}$. The asymptotic equivalences hold because of the definition of $R_{p, n}$ given in \eqref{rpn} depending on the size of $\mathcal{S}$ and the last inequalities by $\lambda_\infty = c_1\log^{1/2}(n)$ and $\lambda_{2} = c_2 \sqrt{p + \log n}$.
\end{enumerate}
Due to the assumption given on the distance between two consecutive change-points, by case 2 above, once the algorithm is terminated, each true change-point $\eta_\ell$ must have an estimated change-point within the distance of $C'r^\ell_{p, n}$ where $C'$ is a constant. If there were two, the more remote one gets removed by case 1 as it is not the closest one. Consequently, each true change-point $\eta_\ell$ must have an estimated change-point within the distance of $C'r_{p,n}^\ell$, which contradicts the initial assumption that there is at least one $\ell \in \{1, \dots, N\}$ without an estimated change-point in $\tilde{\boldsymbol{f}}$ within distance of $r_{p,n}^\ell$ of $\eta_\ell$. This contradiction with the upper bound $\|\dbtilde{\b f} - \boldsymbol{f}\|_{p,n}^2 \leq \frac{C}{4} R_{p,n}$ via \eqref{A14} completes the proof.

\section{Extension to dependent non-Gaussian noise} \label{appendixB}
In this section, we extend our bottom-up methodology to more realistic settings where the noise $\{\ve_{i, t}\}_{i=1, \ldots, p, \; t=1, \ldots, n}$ in \eqref{model} has temporal dependence over $t$ for all $i$ and/or $\b{\ve}_{i, \cdot} = (\ve_{i, 1}, \ldots, \ve_{i, n})^\top$ has non-Gaussianity for all $i$. The specified settings are precisely defined later in this section. We will define larger thresholds to control the behavior of those sums $\sum_{t=t_1}^{t_2}\ve_{i, t}$ for dependent non-Gaussian data. The estimator of $\mu_i^{(j, k)}$ for $j \geq 1$ is obtained by applying the same pruning rule with different thresholds as follows.
\begin{equation} \label{muhat_alphamixing}
\begin{split}
\hat{\mu}_i^{(j, k)} = C^{(j, k)}_{i} \cdot \mathbb{I} \; \Bigg\{ \, & \exists (j', k') \in \mathcal{O}_{j, k} \quad \max_i\big|C_i^{(j', k')}\big| > \lambda_\infty^* \\
& \text{ or } \bigg\{\sum_{i=1}^p \big(C_i^{(j', k')}\big)^2 \bigg\}^{1/2} > \lambda_2^* \, \Bigg\}, \; i=1,\ldots, p,
\end{split}
\end{equation}
where $\mathbb{I}$ is an indicator function and $\mathcal{O}_{j, k}$ is defined as in \eqref{e5.2.11}.


\subsection{Theoretical behavior of the estimators} \label{nonG_error}

\begin{Thm} \label{thm1.dep}
\normalfont  Let the distribution of $\varepsilon_{i, t}$ in model~(1) of the main article be as follows:
\begin{enumerate}[label=(\alph*)]
    \item For all $i=1, \ldots, p$, $\{\varepsilon_{i, t}\}_t$ has mean zero and satisfies Cramer's conditions that
\begin{equation*}
    E|\varepsilon_{i, t}|^k \leq c^{k-2} k! E(\varepsilon_{i, t}^2) < \infty, \quad t=1, \ldots, n, \quad k=3, 4, \ldots,
\end{equation*}
where $c>0$.
\item For all $i=1, \ldots, p$, $\{\varepsilon_{i, t}\}_t$ is $\alpha$-mixing with $\alpha(k) \leq c \nu^k$, $c>0$, $0<\nu<1$.
\end{enumerate}

Let the thresholds satisfy $\lambda^*_\infty = c_1 \log^2(n)$ and $\lambda^*_2 = c_2 \sqrt{p}\log^2(n)$ with sufficiently large constants $c_1 > 0$ and $c_2 > 0$. Let the dimension $p$ satisfy $p \sim n^\alpha$ for some fixed $\alpha \in (0, \infty)$. On the set $D_{p, n}$ defined by
\begin{align*} 
    D_{p, n}=\bigg\{ \forall i = 1, \ldots, p \; \text{ and } \; \forall 1 \leq u \leq v \leq w \leq n,\quad \frac{1}{\sqrt{w-u+1}} \; \Bigg| \sum_{t=u}^w \varepsilon_{i, t} \Bigg| \leq c_4 \log^2(n) \bigg\} \numberthis \label{atl}
\end{align*}
which satisfies $P(D_{p, n}) \rightarrow 1$ as $n \rightarrow \infty$ for a sufficiently large constant $c_4>0$, then we have
\begin{equation} \label{ethm1}
\| {\tilde{\b f}}-\b{f} \|_{p, n}^2 \; \leq \; \frac{1}{n} \Bigg[\log^4(n)\min(c_1^2, c_2^2) + 4N \biggl\lceil\frac{\log n}{\log (1/(1-\rho))}\biggr \rceil 
\log^4(n) \max_\ell \bigg\{c_1^2 \frac{\mathcal{S}_\ell}{p} \wedge c_2^2 \bigg\} \Bigg].
\end{equation}
\end{Thm}
The consistency rate in \eqref{ethm1} differs only by a logarithmic factor from the one obtained under iid Gaussian noise in the main article.

\textbf{Proof.} 
We first show that $P(D_{p, n}) \rightarrow 1$ as $n \rightarrow \infty$. For this, we consider the single sum $\sum_{t=1}^a \varepsilon_{i, t}$ for any $i \in \{1, 2, \ldots, p\}$. From Theorem 1.6 in \citet{bosq1998nonparametric}, we have
\begin{align*}
    P\Bigg(\frac{1}{\sqrt{a}}\Bigg|\sum_{t=1}^a \varepsilon_{i, t} \Bigg| > c\log^2(n) \Bigg) = & O\Big[\log^2(n) \exp{\Big
    \{-c'\log^2(n)\Big\}}\Big] \\ 
    & + O\Big[n \exp{\big \{-c''\log^2(n)\big\}}\Big] \numberthis \label{bosq2},    
\end{align*}
by setting $\varepsilon = \frac{\log_2(n)\log(n)}{\sqrt{a}}$ and $q=\bigg[\frac{a}{\log_2(n)\log(n)} + 1 \bigg]$ as in Theorem 1.6-(2) in \citet{bosq1998nonparametric}. Since there exist up to $n^2$ forms of $\sum_{t=1}^a \varepsilon_{i, t}$ for each $i\in \{1, 2, \ldots, p\}$, by multiplying $pn^2$ to the right-hand side of \eqref{bosq2} via Bonferroni correction, we have $P(D_{p, n}) \rightarrow 1$ as $n \rightarrow \infty$.

Under the null, on $D_{p, n}$, we have
\begin{align*}
    P\Big\{ \forall 0 \leq u < v < w \leq n \quad \max_{i} \; |C_{i; u, v, w}| \leq \lambda^*_{\infty} \Big\} \rightarrow 1,
\end{align*}
as $n \rightarrow \infty$, where $\lambda^*_\infty = c_1 \log^2(n)$ and $C_{i; u, v, w} = \langle \psi_{u, v, w}, \bve_{i, \cdot} \rangle$. This can be obtained directly from the set $D_{p, n}$ using the fact that
\begin{align*}
    |C_{i; u, v, w}| &= \Bigg| \Big\{\frac{w-v}{w-u}\Big\}^{1/2} \frac{\sum_{t=u+1}^v \varepsilon_{i, t}}{\sqrt{v-u} } - \Big\{\frac{v-u}{w-u}\Big\}^{1/2} \frac{\sum_{t=v+1}^w \varepsilon_{i, t}}{\sqrt{w-v} }\Bigg|\\
    &\leq c_4 \log^2(n) \frac{\sqrt{w-v} + \sqrt{v-u}}{\sqrt{w-u}} \numberthis \label{cbound},
\end{align*}
for all $i\in\{1, 2, \ldots, p\}$. This completes $L_\infty$ control of the CUSUM statistics under the null. 

Similarly, under the null on $D_{p, n}$, we have the $L_2$ control of the CUSUM statistics as follows
\begin{align*}
P\bigg\{ \forall 0 \leq u < v < w \leq n \quad \bigg\{ \sum_{i=1}^p \; |C_{i; u, v, w}|^2 \bigg\}^{1/2}  \leq \lambda_{2}^* \bigg\} \rightarrow 1, 
\end{align*}
as $n \rightarrow \infty$, where $\lambda^*_2 = c_2 \sqrt{p}\log^2(n)$.
This can be obtained from Bonferroni correction and the fact that
\begin{align*}
    \Bigg(\sum_{i=1}^p|C_{i; u, v, w}|^2\Bigg)^{1/2} \leq \Big[(c'_4)^2p \log^4(n) \Big]^{1/2},
\end{align*}
obtained from \eqref{cbound}, where $c'_4>0$ is a constant.

The consistency rate in \eqref{ethm1} can be obtained by following the same arguments used in the proof of Theorem \ref{thm1} with different thresholds. We can also obtain results equivalent to Theorems \ref{thm2} and \ref{thm3} by following exactly the same logic with the updated consistency rate in Theorem \ref{thm1.dep}. In conclusion, when the noise satisfies Cramer's conditions with exponential $\alpha$-mixing, larger thresholds are required, which affects the consistency rate in the theorems.

\section{Additional simulations} \label{extrasim}

\begin{table}[ht]
\centering
\resizebox{\textwidth}{!}{
\begin{tabular}{l|lcc}
  \hline
scenario &  sparsity & ($\mathcal{S}_1, \ldots, \mathcal{S}_N$) & ($\theta_1, \ldots, \theta_N$) \\
 \hline
\multirow{5}*{\parbox{4.5cm}{ (Freq) Frequent \\$n=150, p=50, N=4$ \\ $\b{\eta}=30 \cdot (1, 2, 3, 4)$}} & sparse & (1, 1, 1, 1) & (3, 3, 3, 3)\\ 
& dense & (35, 35, 35, 35) & (3.5, 3.5, 3.5, 3.5)\\
& mixed1 & (1, 7, 1, 7) & (3, 3, 3, 3)\\
& mixed2 & (1, 35, 35, 1) & (2.3, 2.8, 2.8, 2.3)\\
& mixed3 & (35, 1, 7, 35) & (2.8, 2.3, 2.5, 2.8)\\
\cline{1-4}
\multirow{3}*{\parbox{4.5cm}{(L.Freq) Less frequent \\$n=300, p=50, N=2$ \\ $\b{\eta}=(100, 200)$}} & sparse & (1, 1) & (1.5, 1.5)\\
& dense & (35, 35) & (1.5, 1.5)\\
& mixed & (1, 35) & (1.5, 1.5)\\
\hline 
\end{tabular}}
\caption{Parameter settings with different types of sparsity for scenarios (Freq) and (L.Freq).}
\label{scenarios_1}
\end{table}


\begin{table}[htbp] 
\centering
\resizebox{\textwidth}{!}{
\begin{tabular}{crrrrrrrrrrr}
  \hline
  & & \multicolumn{7}{c}{$\hat{N}-N$} & & & \\
  \cline{3-9}
 Scenario & Method & $\leq$-3 & -2 & -1 & 0 & 1 & 2 & $\geq$3 & MSE & ARI & time\\ 
  \hline
\multirow{9}*{(Freq)-sparse} 
& TD($L_2$) & 0 & 0 & 6 & 91 & 3 & 0 & 0 & 0.09 & 0.944 & 0.053 \\ 
   & TD($L_\infty$) & 0 & 0 & 0 & 96 & 4 & 0 & 0 & 0.04 & 0.965 & 0.072 \\ 
   & BU($L_2$) & 0 & 0 & 8 & 92 & 0 & 0 & 0 & 0.08 & 0.893 & 0.019 \\ 
   & BU($L_\infty$) & 0 & 0 & 0 & 99 & 1 & 0 & 0 & 0.01 & 0.917 & 0.018 \\ 
   & SBS & 2 & 59 & 32 & 6 & 1 & 0 & 0 & 2.87 & 0.572 & 0.056 \\ 
   & DC & 0 & 0 & 0 & 95 & 5 & 0 & 0 & 0.05 & 0.965 & 0.382 \\ 
   & INSPECT & 0 & 0 & 2 & 95 & 3 & 0 & 0 & 0.05 & 0.959 & 0.118 \\ 
   & scanEH & 0 & 0 & 0 & 95 & 5 & 0 & 0 & 0.05 & 0.966 & 0.380 \\ 
   & BUHDA & 0 & 0 & 0 & 98 & 2 & 0 & 0 & 0.02 & 0.919 & 0.021 \\ 
   \hline
  \multirow{9}*{(Freq)-dense} 
  & TD($L_2$) & 0 & 0 & 0 & 98 & 2 & 0 & 0 & 0.02 & 0.975 & 0.053 \\ 
   & TD($L_\infty$) & 14 & 80 & 6 & 0 & 0 & 0 & 0 & 4.52 & 0.504 & 0.067 \\ 
   & BU($L_2$) & 0 & 0 & 0 & 99 & 1 & 0 & 0 & 0.01 & 0.934 & 0.019 \\ 
   & BU($L_\infty$) & 43 & 57 & 0 & 0 & 0 & 0 & 0 & 6.15 & 0.464 & 0.017 \\ 
   & SBS & 0 & 0 & 2 & 95 & 3 & 0 & 0 & 0.05 & 0.937 & 0.059 \\ 
   & DC & 0 & 98 & 2 & 0 & 0 & 0 & 0 & 3.94 & 0.599 & 0.364 \\ 
   & INSPECT & 0 & 0 & 0 & 99 & 1 & 0 & 0 & 0.01 & 0.960 & 0.120 \\ 
   & scanEH & 0 & 0 & 4 & 96 & 0 & 0 & 0 & 0.04 & 0.959 & 0.376 \\ 
   & BUHDA & 0 & 0 & 1 & 91 & 8 & 0 & 0 & 0.09 & 0.894 & 0.021 \\ 
  \hline
\end{tabular}}
\caption {Distribution of $\hat{N}-N$ for scenarios (Freq)-sparse and (Freq)-dense and all methods over 100 simulation runs. Also the average MSE (Mean Squared Error) of the number of change-points detected, the average Adjusted Rand index (ARI) of the estimated segmentation against the true one, the average computational time in seconds using an Intel Core i9 3.6 GHz CPU with 40 GB of RAM, all over 100 simulations.} \label{table:Freq1}
\end{table}

In the additional simulations, we explore two sets of scenarios, frequent (Freq) change-point settings and less frequent (L.Freq) change-point settings, where the details can be found in Table \ref{scenarios_1}. SN is not included due to its difference in platform. As done in the main paper, for each set of scenarios, different types of sparsity are examined. 

Tables \ref{table:Freq1} and \ref{table:Freq2} summarize the results for frequent (Freq) change-point settings and have similar interpretations to those for Tables \ref{table:ld1} and \ref{table:ld2}. Table \ref{table:LFreq} shows the results of less frequent (L.Freq) change-point settings. BUHDA, scanEH and INSPECT show competitive performance across all settings. As expected, $L_2$ aggregation performs better than $L_\infty$ for both top-down and bottom-up methods in the dense and mixed cases.


\begin{table}[htbp] 
\centering
\resizebox{\textwidth}{!}{
\begin{tabular}{crrrrrrrrrrr}
  \hline
  & & \multicolumn{7}{c}{$\hat{N}-N$} & & & \\
  \cline{3-9}
 Scenario & Method & $\leq$-3 & -2 & -1 & 0 & 1 & 2 & $\geq$3 & MSE & ARI & time\\ 
  \hline
  \multirow{9}*{(Freq)-mixed1}
  & TD($L_2$) & 0 & 0 & 3 & 93 & 4 & 0 & 0 & 0.07 & 0.951 & 0.053 \\ 
   & TD($L_\infty$) & 0 & 39 & 41 & 20 & 0 & 0 & 0 & 1.97 & 0.724 & 0.071 \\ 
   & BU($L_2$) & 0 & 0 & 4 & 96 & 0 & 0 & 0 & 0.04 & 0.900 & 0.018 \\ 
   & BU($L_\infty$) & 0 & 59 & 32 & 9 & 0 & 0 & 0 & 2.68 & 0.656 & 0.018 \\ 
   & SBS & 0 & 57 & 39 & 4 & 0 & 0 & 0 & 2.67 & 0.625 & 0.056 \\ 
   & DC & 0 & 28 & 43 & 29 & 0 & 0 & 0 & 1.55 & 0.762 & 0.367 \\ 
   & INSPECT & 0 & 0 & 0 & 97 & 3 & 0 & 0 & 0.03 & 0.959 & 0.122 \\ 
   & scanEH & 0 & 0 & 0 & 99 & 1 & 0 & 0 & 0.01 & 0.975 & 0.378 \\ 
   & BUHDA & 0 & 0 & 3 & 96 & 1 & 0 & 0 & 0.04 & 0.900 & 0.021 \\ 
  \hline
  \multirow{9}*{(Freq)-mixed2} 
  & TD($L_2$) & 0 & 5 & 29 & 65 & 1 & 0 & 0 & 0.50 & 0.861 & 0.051 \\ 
   & TD($L_\infty$) & 0 & 35 & 64 & 1 & 0 & 0 & 0 & 2.04 & 0.622 & 0.082 \\ 
   & BU($L_2$) & 0 & 13 & 34 & 53 & 0 & 0 & 0 & 0.86 & 0.787 & 0.019 \\ 
   & BU($L_\infty$) & 0 & 75 & 25 & 0 & 0 & 0 & 0 & 3.25 & 0.512 & 0.018 \\ 
   & SBS & 0 & 75 & 24 & 1 & 0 & 0 & 0 & 3.24 & 0.587 & 0.057 \\ 
   & DC & 0 & 44 & 56 & 0 & 0 & 0 & 0 & 2.32 & 0.616 & 0.458 \\ 
   & INSPECT & 0 & 5 & 21 & 71 & 2 & 1 & 0 & 0.47 & 0.864 & 0.116 \\ 
   & scanEH & 0 & 1 & 44 & 55 & 0 & 0 & 0 & 0.48 & 0.845 & 0.403 \\ 
   & BUHDA & 0 & 0 & 12 & 84 & 4 & 0 & 0 & 0.16 & 0.838 & 0.021 \\ 
  \hline
  \multirow{9}*{(Freq)-mixed3} 
  & TD($L_2$) & 0 & 0 & 27 & 72 & 1 & 0 & 0 & 0.28 & 0.901 & 0.051 \\ 
   & TD($L_\infty$) & 10 & 88 & 2 & 0 & 0 & 0 & 0 & 4.44 & 0.551 & 0.064 \\ 
   & BU($L_2$) & 0 & 1 & 34 & 65 & 0 & 0 & 0 & 0.38 & 0.828 & 0.018 \\ 
   & BU($L_\infty$) & 31 & 68 & 1 & 0 & 0 & 0 & 0 & 5.52 & 0.494 & 0.016 \\ 
   & SBS & 0 & 7 & 80 & 13 & 0 & 0 & 0 & 1.08 & 0.720 & 0.068 \\ 
   & DC & 2 & 95 & 3 & 0 & 0 & 0 & 0 & 4.01 & 0.579 & 0.356 \\ 
   & INSPECT & 0 & 0 & 21 & 79 & 0 & 0 & 0 & 0.21 & 0.881 & 0.119 \\ 
   & scanEH & 0 & 12 & 45 & 43 & 0 & 0 & 0 & 0.93 & 0.817 & 0.362 \\ 
   & BUHDA & 0 & 1 & 20 & 76 & 3 & 0 & 0 & 0.27 & 0.830 & 0.021 \\ 
   \hline
\end{tabular}}
\caption {Distribution of $\hat{N}-N$ for scenario (Freq)-mixed1--3 and all methods over 100 simulation runs. Also the average MSE (Mean Squared Error) of the number of change-points detected, the average Adjusted Rand index (ARI) of the estimated segmentation against the true one, the average computational time in seconds using an Intel Core i9 3.6 GHz CPU with 40 GB of RAM, all over 100 simulations.} \label{table:Freq2}
\end{table}


\begin{table}[h!] 
\centering
\resizebox{\textwidth}{!}{
\begin{tabular}{crrrrrrrrrrr}
  \hline
  & & \multicolumn{7}{c}{$\hat{N}-N$} & & &\\
  \cline{3-9}
 Scenario & Method & $\leq$-3 & -2 & -1 & 0 & 1 & 2 & $\geq$3 & MSE & ARI & time\\ 
  \hline
\multirow{9}*{(L.Freq)-sparse} 
  & TD($L_2$) & 0 & 0 & 0 & 98 & 2 & 0 & 0 & 0.02 & 0.985 & 0.182 \\ 
   & TD($L_\infty$) & 0 & 0 & 0 & 84 & 16 & 0 & 0 & 0.16 & 0.954 & 0.245 \\ 
   & BU($L_2$) & 0 & 0 & 0 & 98 & 2 & 0 & 0 & 0.02 & 0.958 & 0.030 \\ 
   & BU($L_\infty$) & 0 & 0 & 0 & 100 & 0 & 0 & 0 & 0.00 & 0.947 & 0.032 \\ 
   & SBS & 0 & 0 & 0 & 74 & 26 & 0 & 0 & 0.26 & 0.922 & 0.220 \\ 
   & DC & 0 & 0 & 0 & 97 & 3 & 0 & 0 & 0.03 & 0.980 & 1.649 \\ 
   & INSPECT & 0 & 0 & 0 & 98 & 2 & 0 & 0 & 0.02 & 0.986 & 0.395 \\ 
   & scanEH & 0 & 0 & 0 & 98 & 2 & 0 & 0 & 0.02 & 0.984 & 1.665 \\ 
   & BUHDA & 0 & 0 & 0 & 99 & 1 & 0 & 0 & 0.01 & 0.953 & 0.034 \\ 
  \hline
\multirow{9}*{(L.Freq)-dense} 
  & TD($L_2$) & 0 & 0 & 0 & 99 & 1 & 0 & 0 & 0.01 & 0.987 & 0.176 \\ 
   & TD($L_\infty$) & 0 & 0 & 60 & 40 & 0 & 0 & 0 & 0.60 & 0.646 & 0.225 \\ 
   & BU($L_2$) & 0 & 0 & 0 & 99 & 1 & 0 & 0 & 0.01 & 0.955 & 0.028 \\ 
   & BU($L_\infty$) & 0 & 0 & 95 & 5 & 0 & 0 & 0 & 0.95 & 0.557 & 0.029 \\ 
   & SBS & 0 & 0 & 0 & 96 & 4 & 0 & 0 & 0.04 & 0.974 & 0.220 \\ 
   & DC & 0 & 0 & 66 & 33 & 1 & 0 & 0 & 0.67 & 0.694 & 1.582 \\ 
   & INSPECT & 0 & 0 & 0 & 100 & 0 & 0 & 0 & 0.00 & 0.981 & 0.394 \\ 
   & scanEH & 0 & 0 & 0 & 99 & 1 & 0 & 0 & 0.01 & 0.983 & 1.654 \\ 
   & BUHDA & 0 & 0 & 0 & 95 & 5 & 0 & 0 & 0.05 & 0.936 & 0.036 \\ 
  \hline
\multirow{9}*{(L.Freq)-mixed} 
& TD($L_2$) & 0 & 0 & 0 & 99 & 1 & 0 & 0 & 0.01 & 0.990 & 0.176 \\ 
   & TD($L_\infty$) & 0 & 0 & 38 & 62 & 0 & 0 & 0 & 0.38 & 0.768 & 0.237 \\ 
   & BU($L_2$) & 0 & 0 & 0 & 99 & 1 & 0 & 0 & 0.01 & 0.958 & 0.029 \\ 
   & BU($L_\infty$) & 0 & 0 & 79 & 21 & 0 & 0 & 0 & 0.79 & 0.622 & 0.030 \\ 
   & SBS & 0 & 0 & 0 & 90 & 10 & 0 & 0 & 0.10 & 0.954 & 0.219 \\ 
   & DC & 0 & 0 & 39 & 60 & 1 & 0 & 0 & 0.40 & 0.813 & 1.613 \\ 
   & INSPECT & 0 & 0 & 0 & 99 & 1 & 0 & 0 & 0.01 & 0.986 & 0.396 \\ 
   & scanEH & 0 & 0 & 0 & 99 & 1 & 0 & 0 & 0.01 & 0.990 & 1.668 \\ 
   & BUHDA & 0 & 0 & 0 & 96 & 4 & 0 & 0 & 0.04 & 0.949 & 0.036 \\ 
   \hline
\end{tabular}}
\caption {Distribution of $\hat{N}-N$ for scenario (L.Freq) and all methods over 100 simulation runs. Also the average MSE (Mean Squared Error) of the number of change-points detected, the average Adjusted Rand index (ARI) of the estimated segmentation against the true one, the average computational time in seconds using an Intel Core i9 3.6 GHz CPU with 40 GB of RAM, all over 100 simulations.} \label{table:LFreq}
\end{table}

\section{Connections between bottom-up tree construction and wavelet transform}\label{wt}


Constructing a bottom-up tree via Algorithm \ref{algo_bu} is the same as applying a conditionally orthonormal wavelet transform to the data matrix. \citet{fryzlewicz2017tail} studies this connection in the univariate setting where $p=1$, and here we extend this to the high-dimensional setting. 
We first define the smooth-type coefficient matrix $\b{S}$. At the initial stage of merging, we assign the initial $\b{S}$ to be the data matrix as below:
\begin{align} \label{initial}
\renewcommand{\arraystretch}{1}
\b{S} =
\begin{pmatrix} 
  s_{1, (0, 1]} &  s_{1, (1, 2]} &  \ldots & s_{1, (n-1, n]} \\
 s_{2, (0, 1]} & s_{2, (1, 2]} &  \ldots & s_{2, (n-1, n]} \\
 \vdots & \vdots & \vdots  & \vdots  \\
 s_{p, (0, 1]} &  s_{p, (1, 2]} &  \ldots & s_{p, (n-1, n]} \\
\end{pmatrix}_{p \times n} 
= \hskip 1em \begin{pmatrix} 
X_{1, 1} & X_{1, 2} & \ldots & X_{1, n} \\ X_{2, 1} & X_{2, 2} & \ldots & X_{2, n} \\ \vdots & \vdots & \vdots & \vdots \\ X_{p, 1} & X_{p, 2} & \ldots & X_{p, n}\\
\end{pmatrix}_{p \times n}.
\end{align}
Each column of $\b{S}$ shows the initial nodes in the current layer, $(0, 1], (1, 2], \ldots, (n-1, n]$. As merges are performed, the set of current layer nodes and the corresponding columns of matrix $\b{S}$ in \eqref{initial} are updated. 


From a wavelet-transform viewpoint, merging a pair of neighboring nodes e.g. $(u, v], (v, w]$ is equivalent to applying a local orthonormal transformation to the scaled mean vector called smooth coefficients as follows:
\begin{align} \label{ss}
\begin{pmatrix} 
s_{i, (u, w]}  \\ 
C_{i; u, v, w} 
\end{pmatrix} = 
\begin{pmatrix} 
a_{u, v, w} & b_{u, v, w}  \\ 
-b_{u, v, w} & a_{u, v, w} \\
\end{pmatrix}
\begin{pmatrix} 
s_{i, (u, v]} \\
s_{i, (v, w]}
\end{pmatrix}, \quad i=1, \ldots, p,
\end{align}
where 
\begin{equation}\label{s}
s_{i, (u, w]} = (w-u)^{-1/2} \sum_{t=u+1}^w X_{i, t}, \quad i=1, \ldots, p,
\end{equation}
and 
\begin{equation}\label{ab}
a_{u, v, w} = \sqrt{(w-v)/(w-u)}, \quad b_{u, v, w} = \sqrt{(v-u)/(w-u)}.
\end{equation}
The transform in \eqref{ss} shows that applying an orthonormal transform updates two smooth-type constant vectors, ($s_{i, (u, v]}, s_{i, (v, w]}$), in $\b{S}$ to one smooth-type vector ($s_{i, (u, w]}$) and one CUSUM statistic vector ($C_{i; u, v, w}$), where $s_{i, (u, w]}$ corresponds to the scaled mean vector of merged segments while $C_{i; u, v, w}$ represents the scaled difference between two segments. In a different view, after performing a merge, a pair of current layer nodes, $(u, v], (v, w]$, are replaced with one top-layer node, $(u, w]$, which reduces the total number of current layer nodes by one. Note that in the wavelet literature, a CUSUM statistic is often called a detail-type coefficient.

Constructing a bottom-up tree in Algorithm \ref{algo_bu} is equivalent to recursively applying the local orthonormal transform in \eqref{ss} to a chosen pair of smooth-type coefficient vectors in the data matrix until only one smooth-type coefficient vector is left (i.e. until only one current layer node is left). 
These transforms produce a data-adaptive multiscale decomposition of the data matrix and eventually convert the input data matrix $\b X$ of dimension $p \times n$ into the matrix containing one column of smooth coefficients and $n-1$ columns of CUSUM statistics as follows: 
\begin{align} \label{transform}
\renewcommand{\arraystretch}{1}
\begin{pmatrix} 
 s_{1, (0, n]} &  {C}_{1; \cdot, \cdot, \cdot} &  \ldots & {C}_{1; \cdot, \cdot, \cdot} \\
 s_{2, (0, n]} & {C}_{2; \cdot, \cdot, \cdot} &  \ldots & {C}_{2; \cdot, \cdot, \cdot} \\
 \vdots & \vdots & \vdots  & \vdots  \\
 s_{p, (0, n]} &  {C}_{p; \cdot, \cdot, \cdot} &  \ldots & {C}_{p; \cdot, \cdot, \cdot} \\
\end{pmatrix}_{p \times n} 
= \hskip 1em \begin{pmatrix} 
X_{1, 1} & X_{1, 2} & \ldots & X_{1, n} \\ X_{2, 1} & X_{2, 2} & \ldots & X_{2, n} \\ \vdots & \vdots & \vdots & \vdots \\ X_{p, 1} & X_{p, 2} & \ldots & X_{p, n}\\
\end{pmatrix}_{p \times n}
 \b\Psi_{n \times n},
\end{align}
where $\b\Psi$ is an orthonormal unbalanced wavelet basis for $\mathbb{R}^n$ obtained in an adaptive way. Note that ${C}_{i; \cdot, \cdot, \cdot}$ in \eqref{transform} is given without exact notation indicating the endpoints of the merged neighboring segments, because they are decided in an adaptive way. The columns of $\b\Psi$ correspond to the wavelet basis used in computing the smooth coefficients and CUSUM statistics as follows:
\begin{align*}
    & s_{i, (0, n]} = \langle (X_{i, 1}, \ldots, X_{i, n})^\top, \psi^{(0)} \rangle, \\
    & {C}_{i; u, v, w} = \langle (X_{i, 1}, \ldots, X_{i, n})^\top, \psi_{u, v, w} \rangle,
\end{align*}
where $\psi^{(0)}$ is the first column vector of $\b\Psi$ and $\psi_{u, v, w}$ is a general form of the remaining column vectors of $\b\Psi$. 

The orthonormality of the unbalanced wavelet basis, $\b\Psi$, implies Parseval's identity:
\begin{equation} \label{pi}
\sum_{t=1}^n (X_{i, t})^2 = \sum_{j=1}^J \sum_{k=1}^{K(j)} \big({C}_{i; \cdot, \cdot, \cdot}^{(j, k)}\big)^2 + (s_{i, (0, n]})^2, \quad \text{ for } i=1, \ldots, p,
\end{equation}
thus
\begin{equation} \label{pi2}
\sum_{t=1}^n (X_{i, t}-\bar{X}_{i})^2 = \sum_{j=1}^J \sum_{k=1}^{K(j)} \big({C}_{i; \cdot, \cdot, \cdot}^{(j, k)}\big)^2, \text{ for } i=1, \ldots, p,
\end{equation}
where $\bar{X}_{i} = 1/n \sum_{t=1}^n X_{i, t}$.
By construction, the CUSUM statistics obtained during early merges tend to be small in magnitude. Thus, Parseval's identity in \eqref{pi2} implies that a large portion of $\sum_{i=1}^p\sum_{t=1}^n (X_{i, t}-\bar{X}_{i})^2$ is explained by only a few large CUSUM statistics arising at coarse levels in a bottom-up tree. In other words, the resulting orthonormal transform of each data sequence tends to encode most of the energy of the signals in only a few CUSUM statistics arising in the later stages of the transform. In this sense, the orthonormal wavelet transform provides a sparse signal representation, which motivates thresholding as the next stage of the algorithm.

\end{document}